\documentclass[12pt, draftclsnofoot, onecolumn]{IEEEtran}
\usepackage{epsfig,latexsym}
\usepackage{float}
\usepackage{indentfirst}
\usepackage{amsmath}
\usepackage{amssymb}
\usepackage{times}
\usepackage{subfigure}
\usepackage{psfrag}
\usepackage{hyperref}
\usepackage{cite}
\usepackage{lastpage}
\usepackage{fancyhdr}
\usepackage{color}
 \usepackage{amsthm}
\usepackage{bigints}
\newtheorem{theorem}{Theorem}
\newtheorem{Lemma}{Lemma}
\newtheorem{lemma}[Lemma]{$\mathbf{Lemma}$}

\begin{document}%%
\title{ {\huge Simple Semi-Grant-Free Transmission Strategies Assisted by Non-Orthogonal Multiple Access}}

\author{ Zhiguo Ding, \IEEEmembership{Senior Member, IEEE}, Robert Schober, \IEEEmembership{Fellow, IEEE},  Pingzhi Fan,   \IEEEmembership{Fellow, IEEE},   and H. Vincent Poor, \IEEEmembership{Fellow, IEEE} \thanks{

    Z. Ding and H. V. Poor are  with the Department of
Electrical Engineering, Princeton University, Princeton, NJ 08544,
USA. Z. Ding
 is also  with the School of
Electrical and Electronic Engineering, the University of Manchester, Manchester, UK (email: \href{mailto:zhiguo.ding@manchester.ac.uk}{zhiguo.ding@manchester.ac.uk}, \href{mailto:poor@princeton.edu}{poor@princeton.edu}).
R. Schober is with the Institute for Digital Communications,
Friedrich-Alexander-University Erlangen-Nurnberg (FAU), Germany (email: \href{mailto:robert.schober@fau.de}{robert.schober@fau.de}).
P. Fan is with the Institute of Mobile
Communications, Southwest Jiaotong University, Chengdu, China (email: \href{mailto:pingzhifan@foxmail.com}{pingzhifan@foxmail.com}).

  }\vspace{-2em}}
 \maketitle
\begin{abstract}
Grant-free transmission is an important feature to be supported by future wireless networks since it reduces the signalling overhead caused by conventional grant-based schemes.  However, for grant-free transmission, the number of users admitted to the same channel is not caped, which can lead to a failure of multi-user detection. This paper  proposes non-orthogonal multiple-access (NOMA) assisted semi-grant-free (SGF) transmission, which is a compromise between  grant-free and  grant-based schemes. 
In particular, instead of reserving channels either for grant-based users or grant-free users,   we focus on an SGF communication scenario, where users are admitted to the same channel via a combination  of grant-based   and grant-free protocols. As a result, a channel  reserved by a grant-based user can be shared by    grant-free users, which improves  both connectivity and spectral efficiency.  Two NOMA assisted SGF contention control mechanisms   are developed  to ensure that, with a small amount of signalling overhead, the number of   admitted grant-free users is carefully controlled  and the interference from the grant-free users to the grant-based users is  effectively suppressed.    Analytical results are provided to 
demonstrate that the two proposed SGF mechanisms employing different successive interference cancelation decoding orders are applicable to different practical network scenarios.  

%identify   the criteria for the selection of the two proposed SGF mechanisms as well as  appropriate successive interference cancelation decoding orders. 

\end{abstract} \vspace{-1em}

\section{Introduction}
Non-orthogonal multiple access (NOMA) has been recently recognized as a promising solution to realize the three key performance requirements  of next-generation  mobile networks, namely enhanced Mobile Broadband (eMBB), 
 Ultra Reliable Low Latency Communications (URLLC), and 
 massive Machine Type Communications (mMTC) \cite{7894280, nomama,jsacnomaxmine}. For example,  existing  studies have demonstrated  that the use of NOMA can significantly improve the system throughput for downlink and uplink transmission without consuming extra bandwidth, which is particularly important for eMBB \cite{6933459,7557079,7812683}. Since NOMA ensures that multiple users can be served in the same time/frequency resource, the users do not have to wait  for serving  even if there are not sufficient  orthogonal resource blocks available, and hence the latency experienced by the users   is   reduced, which is a useful feature for the support of URLLC \cite{6871674he,8345745,Bennisurllc}. The key challenge in realizing  mMTC is the support of massive connectivity, given the  scarce  bandwidth resources, for which NOMA is a perfect solution as it encourages users to share their bandwidth resources instead of solely occupying them \cite{huawei,8395153,8352626}.  While   the application of NOMA to eMBB has been extensively studied, there are few works on   the application of NOMA to URLLC and mMTC.

%Grant-free transmission can be viewed as a practical implementation of NOMA and is particularly important to mMTC and URLLC. 

This paper    focuses  on the application of NOMA to grant-free transmission which is an important feature to be included in mMTC and URLLC. The basic idea of grant-free transmission is that a user is encouraged to transmit whenever it has data to send, without getting a grant from the base station. Therefore,  the lengthy handshaking process between the user and the base station is avoided and  the associated signalling overhead is reduced. The key challenge for grant-free transmission is contention, as multiple users may  choose the same channel to transmit at the same time. To resolve this problem,  two types of grant-free solutions have been proposed. One is to exploit spatial degrees of freedom by applying massive multiple-input multiple-output (MIMO)    to resolve contention \cite{8323218,8320821}. The other is to apply the NOMA principle and use sophisticated multi-user detection (MUD) methods, such as parallel interference cancellation (PIC) and compressed  sensing \cite{8482464,7972955,jsacnoma10}. In general, these grant-free schemes can be viewed as special cases of random access, where the base station does not play any role in multiple access, similar to computer networks. As a result, there is no centralized control for the number of users participating in contention, which means that these grant-free protocols fail  if there is an excessive number of  active users/devices, a likely situation for mMTC and URLLC. %In addition,  most grant-free transmission protocols  require   blind user activity identification and blind channel estimation, which are difficult to   implement in practice due to their high computational complexity and residual ambiguity \cite{8323218,8419284,Dingtsp08}.  

The aim of this paper is to design    NOMA assisted semi-grant-free (SGF) transmission schemes, which can be viewed as a compromise between  grant-free transmission and  conventional grant-based schemes. In particular, instead of reserving channels either for grant-based users or grant-free users,  we focus on an SGF communication scenario   in this paper, where  one user is admitted to a channel via a conventional grant-based protocol and the other users are admitted  to     the same channel in an opportunistic and grant-free manner.   User connectivity can be improved  by considering this   scenario with a combination  of grant-based and grant free protocols,  as   all   channels   in the network are opened up for grant-free transmission, even if they have  been  reserved by  grant-based users.  In order to guarantee that the quality of service (QoS) requirements of the grant-based users are met, the contention among the opportunistic grant-free  users needs to be carefully controlled to ensure that the grant-free users   transmit only if they do not cause too much performance degradation to the grant-based users. Note that such contention control is not possible with pure grant-free protocols \cite{8323218,8320821,8482464,7972955,jsacnoma10}. 

Two SGF schemes for contention control are proposed in this paper, where, unlike pure grant-free transmission,   multiple access is still   controlled by the base station, but with lower signalling overhead compared to the grant-based case.  
For the first   proposed SGF scheme, the base station controls  multiple access by broadcasting a threshold value  to all users and also provides a criterion for the users to determine if they are qualified for transmission, a strategy similar to random beamforming \cite{viswanath02,Zhiguo_mmwave}. As a result,  the  contention control is realized   in an {\it open-loop} manner, which does not require  the users' channel state information (CSI) to be known  at the base station prior to transmission. Compared to grant-based transmission, less signalling overhead is introduced  by SGF since all   users which satisfy the criterion set by the base station are allowed  to  transmit immediately, without going through    individual handshaking processes. Compared to grant-free transmission, in SGF, the number of the users admitted to the same channel  can be carefully controlled, which helps avoid  MUD failure  due to an excessive number of admitted users.       For the second   proposed SGF scheme, {\it distributed contention control} is applied, where a fixed number of users with favourable  channel conditions is granted access. We note that, for the proposed open-loop SGF scheme, the number of   users admitted to the same channel is random, similar to conventional grant-free transmission. In other words,    open-loop SGF may still admit more users to the same channel than   can be supported, whereas  the proposed SGF scheme    with distributed contention control   ensures that   a fixed number of users is granted access.  

The  design of the proposed SGF protocols   largely depends on the   decoding order of successive interference cancellation (SIC) at the base station. For example, if the grant-based user's signal is  decoded in the first stage of SIC, in order to reduce the performance degradation experienced by this user, the opportunistic users which are additionally  granted access to the same channel should have weak channel conditions. On the other hand, if the grant-based user's signal is   decoded in the last stage of SIC, strong grant-free users should be granted the right to transmit. The impact of the SIC decoding order on the performance  of the proposed SGF protocols is investigated. Particularly,   for   open-loop SGF, the threshold broadcasted by the base station and the criterion for the users to determine if they are qualified for transmission    are designed based on  the adopted  SIC decoding order, whereas for SGF with distributed contention control, the criterion for user contention  is   adapted to the   SIC decoding order.    
 Analytical results are provided to demonstrate the superior performance of the proposed NOMA assisted SGF protocols and also the impact of different SIC decoding orders  on suitable application scenarios for  the proposed SGF protocols. In particular, if the grant-based user's signal   is decoded in the first stage of SIC, the proposed SGF protocols are suited for the
scenario, where the grant-based user is close to the base station   and the grant-free users are cell-edge users. On the other hand,  if  the grant-based user's signal  is decoded in the last stage of SIC, the proposed SGF protocols are ideally suited for  the scenario, where the grant-based user is a cell-edge user and the grant-free users are close to the base station.

\section{System Model} \label{section system model}
Consider an SGF uplink NOMA  scenario, where multiple users   are admitted to the same channel  via a combination  of grant-based and grant-free protocols. In particular, among these users, assume that there is one  user, denoted by $\textrm{U}_0$, which needs to be served with  high priority. Via a grant-based protocol, $\textrm{U}_0$ is allocated a dedicated orthogonal resource block, denoted by $\text{B}_0$, for its uplink transmission.
In addition, there are $M$  grant-free   users, denoted by $\textrm{U}_m$, $1\leq m\leq M$,  which  do not have time-critical data   and   compete with each other for admission  to $\text{B}_0$ in an opportunistic  manner.     In a typical machine-type communication network, $\textrm{U}_0$ may  be a    sensor    for healthcare monitoring or critical care, and $\textrm{U}_m$, $1\leq m\leq M$,  may be sensors for power meters or environmental monitoring. In this SGF  scenario,   all  channels   in the network are available for grant-free transmission, even if they have  been  reserved by  grant-based users. Hence, massive connectivity can be supported in a spectrally efficient manner. 

\subsection{Assumptions for SGF Protocol Design}
The proposed SGF protocols are designed under the following assumptions:
\begin{itemize}
\item Recall that  $\textrm{U}_0$ is admitted to $\text{B}_0$ by using  a grant-based protocol. It is assumed that via the broadcast signalling during the handshaking process between $\text{U}_0$ and the base station,    $\textrm{U}_0$'s CSI as well as its transmit power, denoted by $P_0$, become available at the base station and at all   grant-free  users in an error-free manner. 

\item Prior to multiple access, each user knows its own CSI perfectly, but the base station does not acquire  the  CSI  of   the grant-free users,   $\textrm{U}_m$, $1\leq m\leq M$, which reduces the signalling overhead.  Via the proposed SGF protocols, a   portion of the $M$    users are granted  access. We assume that there is a sufficient number of orthogonal preambles for the base station to acquire the CSI of the transmitting grant-free   users in order to facilitate   MUD.

\item $\textrm{U}_m$, $1\leq m\leq M$,  are admitted to $\text{B}_0$ under the condition that the target data rate of $\text{U}_0$, denoted by $R_0$, can still  be  achieved  with  high probability, such that the QoS requirements of   $\text{U}_0$ are satisified.  

\item We assume  that all  users' channels, denoted by $h_m$, $0\leq m \leq M$, exhibit  independent and identically distributed (i.i.d.)   quasi-static  Rayleigh fading. The ordered channel gains are denoted by $h_{(m)}$, where $|h_{(1)}|^2\leq \cdots \leq |h_{(M)}|^2$. 
\end{itemize}

\subsection{Low-Overhead Protocols for Contention Control}
A key step for SGF transmission is   low-overhead contention control.  In this paper, we focus on two types of low-overhead protocols, as described  in the following: 
\subsubsection{Open-loop contention control} The base station broadcasts a channel quality threshold $\tau$. The  users decide whether to join the NOMA transmission by comparing their channel gains to $\tau$. A user is admitted if   its channel gain is below or above the threshold  depending on  the  SIC order employed  by the base station,   see Sections \ref{section sic1} and \ref{section sic2}. 

\subsubsection{Distributed contention control}
Distributed contention control has been extensively studied in the contexts of opportunistic carrier sensing \cite{Zhao2005s} and timer-backoff-based sensor selection \cite{Bletsas06, 6334506}. Take the distributed contention control mechanism proposed in \cite{Zhao2005s} as an example, which   selects the user with the strongest (or the weakest) channel   for   channel access. Once the contention time window starts, each user chooses  a backoff $\tau_m$, which is a strictly decreasing (or increasing) function of the user's channel gain. A user  transmits a beacon to the base station after  $\tau_m$ expires, provided  that $\tau_m$ is smaller than the contention  time window.    As such, the user  with the best (or worst)  channel condition waits for the shortest time and hence   identifies  itself to the base station first. This method will be adopted for distributed contention control for the proposed SGF schemes. We note that  more advanced distributed contention control schemes can select multiple strong (or weak) users, and  a user can acquire the other users' CSI    by using the time differences between the transmitted beacons \cite{Zhao2005s}.  

 The design of the proposed SGF schemes  depends on the SIC decoding order, as illustrated in the  following two sections.

\section{Semi-Grant-Free Protocol - Type I}\label{section sic1}
%This protocol is to guarantee $\textrm{U}_0$'s QoS requirements in a probabilistic way. Or the user has a probabilistic QoS requirement, i.e., its outage probability needs to be smaller than $1-0.9999$. 
In this section, we assume  that the message from $\text{U}_0$ is decoded in the first stage of SIC at the base station. 

\subsection{Open-Loop Contention Control}
The base station   broadcasts a threshold $\tau$, which can be interpreted analogous to   metrics  for the interference temperature in cognitive radio networks \cite{Zhiguo_CRconoma}. Unlike grant-free protocols, which grant  all  $M$ users access, here only   users whose channel gains are below this threshold are admitted to $\text{B}_0$. Without  loss of generality, assume that there are $N$ users whose channel gains can satisfy this condition  and share $\text{B}_0$ with $\text{U}_0$.  

In order to simplify notations, the noise power is assumed to be normalized,   and hence the transmit signal-to-noise ratio (SNR) for $\mathrm{U}_0$'s signal is  identical to  the user's transmit power, $P_0$. Furthermore, assume that all  grant-free  users $\mathrm{U}_m$ use the same transmit power,  denoted by $\bar{P}$. 
Therefore, with SIC, the base station can support the following data rates for the $(N+1)$ users:
\begin{align}\label{rate set}
&\left\{ \log\left(1 +\frac{|h_0|^2P_0}{\sum^{N}_{j=1}|h_{(j)}|^2\bar{P}+1}\right),  \log\left(1 +\frac{|h_{(i)}|^2\bar{P}}{\sum^{i-1}_{j=1}|h_{(j)}|^2\bar{P}+1}\right), 1\leq i\leq N\right\}.
\end{align}
The impact of SGF NOMA transmission on the rates of $\text{U}_0$ and the $N$ grant-free  users is studied separately in the following subsections. 
\subsubsection{Impact of SGF transmission on $\text{U}_0$}  
Based on  \eqref{rate set},  the outage probability of $\text{U}_0$ is given by
\begin{align}\label{q111}
\mathbb{P}_0^{\mathrm{I, OL}} =& \sum_{n=1}^{M}\mathbb{P}(N=n)  \underset{Q_1}{\underbrace{\mathbb{P}\left(\left.\log\left(1+\frac{|h_0|^2P_0}{\sum^{n}_{j=1}|h_{(j)}|^2\bar{P}+1}\right)<R_0\right| N=n\right)}}\\\nonumber
&+ \mathbb{P}\left(|h_{(1)}|^2>\tau \right) \mathbb{P}\left(\left.\log\left(1+ |h_0|^2P_0 \right)<R_0\right| N=0\right),
\end{align}
where $\mathbb{P}(N=n)$ denotes the probability  that there are  $n$ grant-free users whose channel gains are smaller  than $\tau$. The following theorem provides a closed-form expression for $\mathbb{P}_0^{\mathrm{I, OL}} $. 
 
 \begin{theorem}\label{theorem1}
 The outage probability of $\text{U}_0$ achieved by the proposed open-loop Type I SGF protocol can be expressed as follows:
\begin{align}\nonumber 
\mathbb{P}_0^{\mathrm{I, OL}}  =& \sum_{n=1}^{M}\frac{M!}{n!(M-n)!}e^{-(M-n)\tau}\left(1-e^{-\tau}\right)^n   \sum^{n}_{p=0} \frac{{n \choose p}(-1)^pe^{-p\tau} }{(1-e^{-\tau})^n } \left( \sum_{l=0}^{n-1} \frac{  \epsilon_0P_0^{-1} \bar{P}e^{-\epsilon_0P_0^{-1}(1+  \bar{P}p\tau) } }{ (1+\epsilon_0P_0^{-1} \bar{P})^{(l+1)} } \right. \\ \label{eqtherom1} & \left.+ e^{-\epsilon_0P_0^{-1}}-e^{-\epsilon_0P_0^{-1}(1+\bar{P}p\tau)}   
\right) +    \left(1 - e^{-\epsilon_0P_0^{-1}}\right),
\end{align}
where $\epsilon_0=2^{R_0}-1$.
 \end{theorem}
 \begin{proof}
 See Appendix \ref{appendix1}.
\end{proof}

%{\bf do I cover all the boundary cases ($M+1$)?}
%%%%%%%%%%%%%
\subsubsection{Asymptotic analysis of $\mathbb{P}_0^{\mathrm{I, OL}} $} In order to obtain some intuition about  $\mathbb{P}_0^{\mathrm{I, OL}} $, an asymptotic analysis is carried out in the following. 

We first consider the case, where   $\bar{P}$ is fixed,  $P_0\rightarrow\infty$ and $\tau\sim \frac{1}{P_0} $.  In this case, $\mathbb{P}_0^{\mathrm{I, OL}} $ can be approximated as follows: 
\begin{align}
\mathbb{P}_0^{\mathrm{I, OL}}  \overset{(a)}{\approx}&     \epsilon_0P_0^{-1} \bar{P}  \sum_{n=1}^{M}\frac{M!}{(n-1)!(M-n)!}e^{-(M-n)\tau}   \sum^{n}_{p=0}  {n \choose p}(-1)^pe^{-p\tau}     
  +    \epsilon_0P_0^{-1}
  \\\nonumber
  =&     \epsilon_0P_0^{-1} \bar{P}  \sum_{n=1}^{M}\frac{M!}{(n-1)!(M-n)!}e^{-(M-n)\tau}  \left(1-e^{-\tau}\right)^n
  +    \epsilon_0P_0^{-1},
\end{align}
where step (a) follows by using the binomial expansion and the series expansion of the exponential function. 
The approximation of $\mathbb{P}_0^{\mathrm{I, OL}} $ can be further simplified as follows: 
\begin{align}\nonumber
\mathbb{P}_0^{\mathrm{I, OL}}  \approx&     \epsilon_0P_0^{-1}+ \epsilon_0P_0^{-1} \bar{P}  \sum_{n=1}^{M}\frac{M!}{(n-1)!(M-n)!\tau^n }e^{-(M-n)\tau} 
  \\ \label{deltax1}
   \approx& \epsilon_0P_0^{-1}+ \underset{\Delta_{\mathbb{P}_0^{\mathrm{I, OL}} }}{\underbrace{ \tau   \epsilon_0P_0^{-1} \bar{P} Me^{-(M-1)\tau}}}   
 .
\end{align}

{\it Remark 1:} Recall that the outage probability of $\text{U}_0$ is degraded by admitting the grant-free users to $\text{B}_0$. The approximation in \eqref{deltax1} clearly shows this degradation, as $\Delta_{\mathbb{P}_0^{\mathrm{I, OL}} }$ is the extra cost for admitting the grant-free users to $\text{B}_0$. However,  this cost goes to zero, if $\bar{P}$ is fixed,  $P_0\rightarrow\infty$, and $\tau\sim \frac{1}{P_0} $. In other words, by carefully choosing $\tau$, $\bar{P}$, and $P_0$, it is possible to ensure that $\text{U}_0$ experiences the same outage probability with the proposed SGF scheme as in the grant-based case. However,   SGF can ensure that more users are connected, compared to a grant-based scheme.

{\it Remark 2:} The importance of considering the case with  large $P_0$ and  small $\bar{P}$ is explained in the following. Recall that $P_0$ and $\bar{P}$ denote the transmit SNRs at $\text{U}_0$ and the grant-free users, respectively. Therefore, in practice, the case with  large $P_0$ and  small $\bar{P}$  represents an important uplink scenario, where the grant-based user is close to the base station and the grant-free users   are at the edge of the cell. The proposed Type I SGF protocol is ideally   suited to this scenario to support massive connectivity in this situation without causing significant performance degradation to the grant-based user. 

Next, we consider a different  case, where    $\tau$   and $\bar{P}$ are fixed,  and $ {P}_0\rightarrow\infty$. 
 In this case, $\mathbb{P}_0^{\mathrm{I, OL}} $ in \eqref{eqtherom1} can be approximated as follows:
 \begin{align}
\mathbb{P}_0^{\mathrm{I, OL}}  \approx& \epsilon_0P_0^{-1} \bar{P}\sum_{n=1}^{M}\frac{M!}{n!(M-n)!}e^{-(M-n)\tau}    \sum^{n}_{p=0}  {n \choose p}(-1)^pe^{-p\tau} \left( n   + p\tau
\right) +     \epsilon_0P_0^{-1} .
\end{align}
The sum of the binomial coefficients in the above equation can be simplified as follows:
\begin{align}
\sum^{n}_{p=0}  {n \choose p}(-1)^pe^{-p\tau}      p  = -ne^{-\tau}\left(1-e^{-\tau}\right)^{n-1}.
\end{align}
Therefore, the approximation of $\mathbb{P}_0^{\mathrm{I, OL}}  $ can be simplified as follows:
 \begin{flalign}  \label{approximation1}
&\mathbb{P}_0^{\mathrm{I, OL}}  \approx \epsilon_0P_0^{-1}+  \underset{\Delta_{\mathbb{P}_0^{\mathrm{I, OL}} }}{\underbrace{\epsilon_0P_0^{-1} \bar{P}\sum_{n=1}^{M}\frac{M!e^{-(M-n)\tau} \left(1-e^{-\tau}\right)^{n-1} \left(1 - e^{-\tau} - \tau e^{-\tau}
\right) }{(n-1)!(M-n)!}}} .
\end{flalign}
{\it Remark 3:} Similar to the case considered    in Remark 1,  one can also observe that in this case the gap between the outage probabilities of the grant-based and SGF schemes, $\Delta_{\mathbb{P}_0^{\mathrm{I, OL}} }$, is reduced to zero as $P_0$ grows. The outage probability  with a constant $\tau$ is worse than that   with $\tau\sim \frac{1}{P_0}$. On the other hand,  SGF with a constant $\tau$ can support more grant-free users than SGF with  $\tau\sim \frac{1}{P_0}$, when $P_0\rightarrow \infty$. Following     steps similar to the ones to obtain \eqref{deltax1} and \eqref{approximation1}, one can show that $\mathbb{P}_0^{\mathrm{I, OL}} $ also  approaches zero  when $\tau$ and ${P}_0$ are fixed but  $\bar{P}\rightarrow0$.   
%%%%%%%%%%%%%%%%%%%%%%%%%%%%%%%%%%%%%%
\subsubsection{Impact of SGF transmission on grant-free users}\label{subsection1}
Without loss of generality,   assume that $N$ grant-free users have been selected by the proposed protocol. 
As shown in \eqref{rate set}, the   proposed Type I SGF protocol can support the following  rates for the   grant-free users:
\begin{align} \label{taxx}
 \log\left(1 +\frac{|h_{(i)}|^2\bar{P}}{\sum^{i-1}_{j=1}|h_{(j)}|^2\bar{P}+1}\right), 
 \end{align}
where $1\leq i\leq N$. Without loss of generality, assume that each grant-free user wants to send $L$ bits to the base station. The probability that  the number of bits sent by    grant-free user $i$ within $\text{B}_0$ is less than $L$ is given by 
\begin{align}\label{Pi}
\mathbb{P}_i^{\mathrm{I, OL}}  = \mathbb{P}\left(\mathrm{B}_0 \log\left(1 +\frac{|h_{(i)}|^2\bar{P}}{\sum^{i-1}_{j=1}|h_{(j)}|^2\bar{P}+1}\right)
<L\right).
\end{align}

{\it Remark 4}: It is assumed that a user can adapt its transmit data rate according to \eqref{taxx}, which requires that one user has access to the other users' CSI. This CSI knowledge  can be obtained by adopting  the beacon-based distributed contention control  scheme in \cite{Zhao2005s}, which can be applied in the open-loop mechanism  not  for contention control, but  for one user to acquire other users' CSI.    However, if this CSI knowledge is not available to the users,   the probability $\mathbb{P}_i^{\mathrm{I, OL}}  $ in \eqref{Pi} can also be viewed as  a lower bound on the outage probability, if each user uses $\breve{R}_i\triangleq \frac{L}{\mathrm{B}_0}$ as its target data rate.   

  The probability $\mathbb{P}_i^{\mathrm{I, OL}} $ can be calculated as follows:
\begin{align} \label{Pidfd}
\mathbb{P}_i^{\mathrm{I, OL}}  =& \mathbb{P}\left( \sum^{i-1}_{j=1}|h_{(j)}|^2>\breve{\epsilon}_i^{-1}|h_{(i)}|^2-\bar{P}^{-1} \right)\\\nonumber=& \mathbb{P}\left( \sum^{i-1}_{j=1}|h_{(j)}|^2>\breve{\epsilon}_i^{-1}|h_{(i)}|^2-\bar{P}^{-1}, |h_{(i)}|^2>\breve{\epsilon}_i\bar{P}^{-1} \right) +\mathbb{P}\left(|h_{(i)}|^2<\breve{\epsilon}_i\bar{P}^{-1} \right),
\end{align}
where $\breve{\epsilon}_i= 2^{\breve{R}_i}-1$.
Note that there is  a hidden constraint $|h_{(i)}|^2<\tau$ since the channel gains of all  selected users are smaller  than $\tau$. 
It is important to point out that $|h_{(i)}|^2 $ and $\sum^{i-1}_{j=1}|h_{(j)}|^2$ are correlated. To evaluate the probability $\mathbb{P}_i^{\mathrm{I, OL}} $, we note that, conditioned on $|h_{(i)}|^2$,  $\sum^{i-1}_{j=1}|h_{(j)}|^2$ can be viewed  as the sum of $(i-1)$ i.i.d. variables, denoted by $\breve{h}_k$,   with the following cumulative distribution function (CDF):
\begin{align}\label{fhj1}
F_{|\breve{h}_{k}|^2}(y)= \left\{ \begin{array}{ll}\frac{1-e^{-y}}{1-e^{-|h_{(i)}|^2}},
& \text{if}\quad y\leq |h_{(i)}|^2 \\1,&\text{if}\quad y>|h_{(i)}|^2
 \end{array} \right..
\end{align} 
Following  steps similar to those in the proof for Theorem \ref{theorem1}, conditioned on $|h_{(i)}|^2$, the probability density function (pdf) of $\sum^{i-1}_{j=1}|h_{(j)}|^2$ can be obtained as follows:
\begin{align}\label{fsumx1}
&f_{\sum^{i-1}_{j=1}|h_{(j)}|^2}(y)  = \sum^{i-1}_{p=0} \frac{{i-1 \choose p}(-1)^pe^{-p|h_{(i)}|^2} }{(1-e^{-|h_{(i)}|^2})^{i-1} (i-2)!}   (y-p|h_{(i)}|^2)^{i-2} e^{-(y-p|h_{(i)}|^2)} u(y-p|h_{(i)}|^2) ,
\end{align}
where $u(\cdot)$ denotes the unit step function. 
On the other hand, $ |h_{(i)}|^2 $ is the $i$-th smallest value  among $N$  i.i.d. random variables following the distribution in \eqref{fhj}, instead of \eqref{fhj1}. By applying order statistics \cite{David03}, the pdf of $ |h_{(i)}|^2 $  can be found as follows:
\begin{align}\label{fsumx2}
f_{|h_{(i)}|^2}(x) = &c_{N,i}\frac{e^{-x}(1-e^{-x})^{i-1} \left(e^{-x}-e^{-\tau}\right)^{N-i}}{(1-e^{-\tau})^N}, \end{align}
 where $c_{N,i}= \frac{N!}{(i-1)!(N-i)!}$.

If $\tau\leq \breve{\epsilon}_i \bar{P}^{-1}$, the probability   is given by
\begin{align} 
\mathbb{P}_i^{\mathrm{I, OL}}  =&\int^{\tau}_{0} f_{|h_{(i)}|^2}(x)dx=1,  \end{align}
otherwise,
\begin{align} \label{Pix1}
\mathbb{P}_i^{\mathrm{I, OL}}  =&\int^{\breve{\epsilon}_i\bar{P}^{-1}}_{0} f_{|h_{(i)}|^2}(x)dx +
\int_{\breve{\epsilon}_i\bar{P}^{-1}}^{\tau}f_{|h_{(i)}|^2}(x)\int^{\infty}_{\breve{\epsilon}_i^{-1}x-\bar{P}^{-1}}  f_{\sum^{i-1}_{j=1}|h_{(j)}|^2}(y)dydx.
\end{align}
It is worth pointing out that  the upper limit  of the integration of the pdf of  $\sum^{i-1}_{j=1}|h_{(j)}|^2$ should be $|h_{(i)}|^2(i-1)$, but it can be replaced by   $\infty$ as shown in Lemma~\ref{lemma1} in Appendix \ref{appendix1}. 

{\it Remark 5:} The probability $\mathbb{P}_i^{\mathrm{I, OL}} $ can be  evaluated numerically by   substituting \eqref{fsumx1} and \eqref{fsumx2} into \eqref{Pix1}, but  a  closed-form expression is difficult to obtain due to   the step function in \eqref{fsumx1}. 

%We need $(i-1)\epsilon_i\leq 1$, otherwise an error floor.
 
\subsection{Distributed Contention Control}
After the base station broadcasts $\tau$,    assume that there are $N$ users  whose channels are worse  than $\tau$. Only these $N$ users are allowed to participate in contention and a fixed number of users will be admitted to $\text{B}_0$ by applying distributed contention control  as discussed at the end of Section \ref{section system model}. Due to space limitations, we will focus on the case that only one user is granted access, i.e.,    the   user with  channel gain $h_{(1)}$ is scheduled if $|h_{(1)}|^2<\tau$. Therefore,   the outage probability of $\text{U}_0$ can be expressed as follows:
\begin{align}\label{poq6}
\mathbb{P}_0^{\mathrm{I, DCC}}  =& \sum_{n=1}^{M}\mathbb{P}\left(|h_{(n)}|^2<\tau, |h_{(n+1)}|^2>\tau\right)   \mathbb{P}\left(\left.\log\left(1+\frac{|h_0|^2P_0}{ |h_{(1)}|^2\bar{P}+1}\right)<R_0\right| N=n\right)\\\nonumber
&+ \mathbb{P}\left(|h_{(1)}|^2>\tau \right) \mathbb{P}\left(\left.\log\left(1+ |h_0|^2P_0 \right)<R_0\right| N=0\right).
\end{align}
The outage probability in \eqref{poq6} can be equivalently expressed as follows:
\begin{align}\label{poq6new}
\mathbb{P}_0^{\mathrm{I, DCC}}  =&   \underset{Q_2}{\underbrace{\mathbb{P}\left( (|h_{(1)}|^2<\tau, \log\left(1+\frac{|h_0|^2P_0}{ |h_{(1)}|^2\bar{P}+1}\right)<R_0 \right)}}\\\nonumber
&+ \mathbb{P}\left(|h_{(1)}|^2>\tau ,\log\left(1+ |h_0|^2P_0 \right)<R_0\right).
\end{align}
With some algebraic manipulations, $Q_2$ can be expressed as follows: 
\begin{align}\label{Q6xs}
Q_2 =&  \left(1-e^{-M\tau}\right)\left(1-e^{-\epsilon_0P_0^{-1}}\right)- e^{-M\tau}\left(e^{-\epsilon_0P^{-1}_0}-e^{-\epsilon_0P_0^{-1}(1+\bar{P}\tau)}\right)   \\\nonumber&+ e^{M\bar{P}^{-1}}\frac{e^{-(1+M\bar{P}^{-1}\epsilon_0^{-1} P_0)\epsilon_0P^{-1}_0 }-e^{-(1+M\bar{P}^{-1}\epsilon_0^{-1} P_0)\epsilon_0P_0^{-1}(1+\bar{P}\tau) }}{1+M\bar{P}^{-1}\epsilon_0^{-1} P_0} .
\end{align}
By substituting \eqref{Q6xs} into \eqref{poq6new} and using the fact that the second probability in \eqref{poq6new} is a product of two simple probabilities, $\mathbb{P}\left(|h_{(1)}|^2>\tau\right)=e^{-M\tau}$ and $\mathbb{P}\left( \log\left(1+ |h_0|^2P_0 \right)<R_0\right)=1-e^{-\epsilon_0P_0^{-1}}$, a closed-form expression for the outage probability of $\text{U}_0$ can be obtained for the case with distributed contention  control.

In order to obtain some insight, consider the case, where $\bar{P}$ and $\tau$ are fixed, $P_0\rightarrow \infty$. In this case, we have 
\begin{align}\label{Q6x1}
Q_2 \approx&   \left(1-e^{-M\tau}\right)\epsilon_0P_0^{-1}- e^{-M\tau} \epsilon_0P_0^{-1}\bar{P}\tau   \\\nonumber&+ e^{M\bar{P}^{-1}}\frac{e^{-(\epsilon_0P_0^{-1}+M\bar{P}^{-1})  }-e^{-(\epsilon_0P_0^{-1}+M\bar{P}^{-1} )(1+\bar{P}\tau) }}{1+M\bar{P}^{-1}\epsilon_0^{-1} P_0} .
\end{align}
 $Q_2$ can be further approximated as follows:
\begin{align}\label{Q6x2}
Q_2 \approx&   \left(1-e^{-M\tau}\right) \epsilon_0P_0^{-1}- e^{-M\tau} \epsilon_0P_0^{-1}\bar{P}\tau   +\frac{1 -\epsilon_0P_0^{-1}-  e^{-M\tau}\left(1- \epsilon_0P_0^{-1} - \epsilon_0P_0^{-1}\bar{P}\tau \right)}{1+M\bar{P}^{-1}\epsilon_0^{-1} P_0}.
\end{align}
Therefore, the outage probability $\mathbb{P}_0^{\mathrm{I, DCC}}$ can be approximated as follows:
\begin{align}\nonumber
\mathbb{P}_0^{\mathrm{I, DCC}}  \approx&   e^{-M\tau}\epsilon_0P_0^{-1}+ \left(1-e^{-M\tau}\right)\epsilon_0P_0^{-1}- e^{-M\tau} \epsilon_0P_0^{-1}\bar{P}\tau  \\\nonumber & +\frac{1 -\epsilon_0P_0^{-1}-  e^{-M\tau}\left(1- \epsilon_0P_0^{-1} - \epsilon_0P_0^{-1}\bar{P}\tau \right)}{1+M\bar{P}^{-1}\epsilon_0^{-1} P_0}\\\label{Q6x3}=&
   \epsilon_0P_0^{-1} - e^{-M\tau} \epsilon_0P_0^{-1}\bar{P}\tau  +\frac{1 -\epsilon_0P_0^{-1}-  e^{-M\tau}\left(1- \epsilon_0P_0^{-1} - \epsilon_0P_0^{-1}\bar{P}\tau \right)}{1+M\bar{P}^{-1}\epsilon_0^{-1} P_0}.
\end{align}
With some algebraic manipulations, we can obtain the following lemma. 
\begin{lemma}\label{lemma2}
Consider the case, where $\bar{P}$ and $\tau$ are fixed, and $P_0\rightarrow \infty$. In this case, $\mathbb{P}_0^{\mathrm{I, DCC}} $ can be approximated as follows:
\begin{align}  \label{idcc}
\mathbb{P}_0   ^{\mathrm{I, DCC}} 
\approx&  
\epsilon_0P_0^{-1}+\underset{\Delta_{\mathbb{P}_0^{\mathrm{I, DCC}} }}{\underbrace{\epsilon_0P_0^{-1}\bar{P}\left( \frac{\left(1-e^{-M\tau}\right)}{M} - e^{-M\tau}\tau\right)}} .
\end{align}
\end{lemma} 

{\it Remark 6:} From the   asymptotic result in \eqref{idcc}, one can observe  that the difference between the outage probabilities for the cases where   the grant-free user is and is not admitted to $\text{B}_0$, $\Delta_{\mathbb{P}_0^{\mathrm{I, DCC}} }$, approaches zero as $M$ approaches infinity. This is different from the previous case considered  in \eqref{approximation1}, where increasing $M$ deteriorates the outage probability of $\text{U}_0$.   In other words, the use of distributed contention control can effectively ensure that the performance of  $\text{U}_0$ is guaranteed even if there are many grant-free users. We note, however, that the use of distributed contention control results in a higher system overhead than the open-loop scheme. 

To analyze the performance of the selected grant-free user,  assume that the number of users whose  channel gains are below the threshold is fixed and  denoted by $N$, $N>0$, i.e., the number of users which are qualified to participate in contention is assumed to be fixed, and the unordered channel gains of these users follow the distribution in  \eqref{fhj}. Therefore,  the outage probability of the selected grant-free user is given by
\begin{align}\label{eq fig4}
\mathbb{P}_1^{\mathrm{I, DCC}}  =& 1- \mathbb{P}\left(\log\left(1+\frac{|h_0|^2P_0}{ |h_{(1)}|^2\bar{P}+1}\right)>R_0, \log\left(1+ |h_{(1)}|^2\bar{P}\right) >R_i \right) \\\nonumber=&
1- \mathbb{P}\left( |h_{(1)}|^2<\bar{P}^{-1}\epsilon_0^{-1}P_0|h_0|^2-\bar{P}^{-1},    |h_{(1)}|^2  > \bar{P}^{-1}\epsilon_i,|h_{(1)}|^2<\tau \right) ,
\end{align}
where   it is assumed that  all  grant-free users  have the same  target data rate, denoted by $R_i$. 
If $\bar{P}^{-1}\epsilon_i\geq \tau$, $\mathbb{P}_1^{\mathrm{I, DCC}}  =1$. This is due to the fact that the user is selected because  its channel gain is   smaller than $\tau$, i.e., $|h_{(1)}|^2<\tau$. On the other hand, $\bar{P}^{-1}\epsilon_i$ is the target SNR of the user since $\log(1+|h_{(1)}|^2\bar{P})>R_i$ means $|h_{(1)}|^2>\bar{P}^{-1}\epsilon_i$.  If $\bar{P}^{-1}\epsilon_i\geq \tau$,  the user's target SNR will never be met.   If $\bar{P}^{-1}\epsilon_i< \tau$, 
the outage probability can be further rewritten as follows:
\begin{align}
\mathbb{P}_1^{\mathrm{I, DCC}}  = &
1- \mathbb{P}\left( \bar{P}^{-1}\epsilon_0^{-1}P_0|h_0|^2-\bar{P}^{-1}>0,  |h_{(1)}|^2<\bar{P}^{-1}\epsilon_0^{-1}P_0|h_0|^2-\bar{P}^{-1},\right.\\\nonumber&\left.   |h_{(1)}|^2  > \bar{P}^{-1}\epsilon_i,|h_{(1)}|^2<\tau \right) \\\nonumber
= &
1- \mathbb{P}\left(  \epsilon_0^{-1}P_0|h_0|^2-1>\bar{P}\tau,   \bar{P}^{-1}\epsilon_i< |h_{(1)}|^2<\tau \right) \\\nonumber &-\mathbb{P}\left( 0<\bar{P}^{-1}\epsilon_0^{-1}P_0|h_0|^2-\bar{P}^{-1}<\tau, \bar{P}^{-1}\epsilon_i<|h_{(1)}|^2<\bar{P}^{-1}\epsilon_0^{-1}P_0|h_0|^2-\bar{P}^{-1}  \right). 
\end{align}
By applying order statistics \cite{David03} and also treating $|h_{(1)}|^2$ as the smallest value among $N$ i.i.d. random variables following      the distribution in  \eqref{fhj}, the outage probability of the selected grant-free user can be expressed as follows: 
\begin{align}\nonumber
\mathbb{P}_1^{\mathrm{I, DCC}}  = & 1- e^{-\epsilon_0P_0^{-1}(1+\bar{P}\tau)}\frac{\left(e^{-\bar{P}^{-1}\epsilon_i}-e^{-\tau}\right)^N}{(1-e^{-\tau})^N} -\frac{\left(e^{-\bar{P}^{-1}\epsilon_i}-e^{-\tau}\right)^N}{(1-e^{-\tau})^N} \left(e^{-\epsilon_0P_0^{-1}}-e^{-\epsilon_0P_0^{-1}(1+\bar{P}\tau)}\right)\\\nonumber&+
\frac{\left(e^{-(\bar{P}^{-1}\epsilon_0^{-1}P_0x-\bar{P}^{-1} )}-e^{-\tau}\right)^N}{(1-e^{-\tau})^N}
  \frac{e^{-(1+k\bar{P}^{-1}\epsilon_0^{-1}P_0)\epsilon_0P_0^{-1} } -e^{-(1+k\bar{P}^{-1}\epsilon_0^{-1}P_0)\epsilon_0P_0^{-1}(1+\bar{P}\tau) }}{(1+k\bar{P}^{-1}\epsilon_0^{-1}P_0)}
\\ 
= & 1-  \frac{\left(e^{-\bar{P}^{-1}\epsilon_i}-e^{-\tau}\right)^N}{(1-e^{-\tau})^N}  e^{-\epsilon_0P_0^{-1}}  +
\sum^{N}_{k=0} \frac{{N\choose k} (-1)^{N-k}e^{k\bar{P}^{-1}-\tau(N-k)} }{(1-e^{-\tau})^N}\\\nonumber
&\times \frac{e^{-(\epsilon_0P_0^{-1} +k\bar{P}^{-1} )} -e^{-(\epsilon_0P_0^{-1}+k\bar{P}^{-1} )(1+\bar{P}\tau) }}{(1+k\bar{P}^{-1}\epsilon_0^{-1}P_0)}.
\end{align}
{\it Remark 7:} An error floor exists for  $\mathbb{P}_1^{\mathrm{I, DCC}}  $, as explained in the following. By increasing $P_0$ and fixing $\bar{P}$, the probability for   event, $\log\left(1+\frac{|h_0|^2P_0}{ |h_{(1)}|^2\bar{P}+1}\right)<R_0$, goes to zero, but not that for   event, $\log\left(1+ |h_{(1)}|^2\bar{P}\right) <R_i $.  By fixing $P_0$ and increasing $\bar{P}$, the probability for   event, $\log\left(1+\frac{|h_0|^2P_0}{ |h_{(1)}|^2\bar{P}+1}\right)<R_0$, approaches one, whereas the probability for   event, $\log\left(1+ |h_{(1)}|^2\bar{P}\right) <R_i $, goes to zero.  Note that  if $P_0=\bar{P}\rightarrow \infty$, the probability can be approximated as follows:
\begin{align}\label{remark 5sx}
\mathbb{P}_1^{\mathrm{I, DCC}}  
\rightarrow &  
\sum^{N}_{k=0} \frac{{N\choose k} (-1)^{N-k}e^{ -\tau(N-k)} }{(1-e^{-\tau})^N} \frac{1 -e^{-(\epsilon_0+k ) \tau }}{(1+k\epsilon_0^{-1})},
\end{align}
which is no longer a function of the transmit powers.

\section{Semi-Grant-Free Protocol - Type II}\label{section sic2}
In this section, we assume  that   $\textrm{U}_0$'s message is decoded  in the last stage of SIC. The corresponding design of the two proposed SGF transmission schemes  is described in the following two subsections. 

\subsection{Open-Loop Contention Control}
Since  $\textrm{U}_0$'s message is decoded after the messages from the grant-free users are decoded at the base station,    grant-free users with strong channel conditions should be scheduled, which means the following change to the contention control mechanism. Upon receiving the threshold $\tau$ from the base station, only   users whose channel gains are above the threshold will participate in the SGF transmission. Assume  again that $N$ users are selected, without loss of generality. 

\subsubsection{Impact of SGF transmission on $\text{U}_0$}
Since $\textrm{U}_0$'s message is decoded in the last stage of SIC, $\textrm{U}_0$'s outage performance   also depends on whether   the messages from the $N$ grant-free users can be correctly decoded at the base station.  Therefore, the outage probability experienced by $\text{U}_0$ can be expressed as follows:
\begin{align}\label{Pxx}
\mathbb{P}_0^{\mathrm{II, OL}}  = &1 - \sum^{M}_{n=1}\mathbb{P}\left(N=n\right)Q_3 -\mathbb{P}\left(N=0\right)  \mathbb{P}\left(\log\left( 1+|h_0|^2P_0 \right)>R_0\right),
\end{align}
where 
\begin{align}
Q_3=&\mathbb{P}\left(\log\left( 1+|h_0|^2P_0 \right)>R_0,  \left. \log\left(1+\frac{\bar{P}\sum^{M}_{j=M-n+1} |h_{(j)}|^2}{1+|h_0|^2P_0}\right)>R_{sum,n}\right|N=n\right),
\end{align}
and   $R_{sum,n}$ is the target sum rate of the $n$ grant-free users. Without loss of generality, again assume that the grant-free users have the same target data rate,   $R_i$, and    $R_{sum,n}=nR_i$. 
It is worth pointing out that a sum-rate based criterion for successful SIC decoding is used in \eqref{Pxx}. Otherwise, the outage events of the $n$ consecutive  SIC stages need to be taken into consideration, which makes the evaluation difficult due to the correlation between the signal-to-interference-plus-noise ratios (SINRs) in the different  SIC stages. 

The following theorem provides a closed-form expression for the outage probability of $\text{U}_0$. 
\begin{theorem}\label{theorem2}
The outage probability of $\text{U}_0$ achieved by the proposed open-loop Type  II SGF protocol can be expressed as follows: 
\begin{align} \label{theorem2 o}
&\mathbb{P}_0^{\mathrm{II, OL}}  = 1 - \sum^{M}_{n=1}\frac{M!}{n!(M-n)!}e^{-n\tau}\left(1-e^{-\tau}\right)^{M-n} \\\nonumber &\times  \left(\sum^{n-1}_{l=1}\frac{\epsilon_{s,n}^l\bar{P}^{-l}P_0^l}{l!}e^{-\tau_n} \frac{\Gamma\left(l+1, (\bar{\tau}_n-\tau_n)(1+\epsilon_{s,n}\bar{P}^{-1}P_0)\right)}{(1+\epsilon_{s,n}\bar{P}^{-1}P_0)^{l+1}}\right.\\\nonumber &\left.+ e^{-\epsilon_0P_0^{-1}}-e^{-\bar{\tau}_n}  + \frac{  e^{- \tau_n- (\bar{\tau}_n-\tau_n)(1+\epsilon_{s,n}\bar{P}^{-1}P_0) }}{1+\epsilon_{s,n}\bar{P}^{-1}P_0}\right) -\left(1-e^{-\tau}\right)^M e^{-\epsilon_0P_0^{-1}},
\end{align}
where $\epsilon_{s,n}=2^{R_{sum,n}}-1$, $\tau_n=\frac{n\tau\epsilon_{s,n}^{-1}\bar{P}-1}{P_0}$, and $\bar{\tau}_n=\max\left(\tau_n, \epsilon_0P_0^{-1}\right)$. 
\end{theorem}
\begin{proof}
See Appendix \ref{prooftheorem2}. 
\end{proof}
In order to obtain   insightful analytical results, an asymptotic study is carried out in the following. When $P_0$ and $\tau$ are fixed and $\bar{P}\rightarrow \infty$, the incomplete Gamma function in \eqref{theorem2 o} can be approximated as follows:
\begin{align}\nonumber
\Gamma\left(l+1, (\bar{\tau}_n-\tau_n)(1+\epsilon_{s,n}\bar{P}^{-1}P_0)\right) &=l!e^{-(\bar{\tau}_n-\tau_n)(1+\epsilon_{s,n}\bar{P}^{-1}P_0)}\sum^{l}_{m=0}\frac{(\bar{\tau}_n-\tau_n)^m(1+\epsilon_{s,n}\bar{P}^{-1}P_0)^m}{m!}\\  \approx&l!e^{-(\bar{\tau}_n-\tau_n) }\sum^{l}_{m=0}\frac{(\bar{\tau}_n-\tau_n)^m }{m!}.
\end{align}
Therefore, the expression for $\mathbb{P}_0^{\mathrm{II, OL}} $ can be approximated as follows:
\begin{align}\nonumber 
\mathbb{P}_0^{\mathrm{II, OL}}  \approx& 1 - \sum^{M}_{n=1}\frac{M!}{n!(M-n)!}e^{-n\tau}\left(1-e^{-\tau}\right)^{M-n}   \left(\sum^{n-1}_{l=1}\epsilon_{s,n}^l\bar{P}^{-l}P_0^l  e^{-\bar{\tau}_n  }\sum^{l}_{m=0}\frac{(\bar{\tau}_n-\tau_n)^m }{m!}  + e^{-\epsilon_0P_0^{-1}} \right)\\\nonumber &-\left(1-e^{-\tau}\right)^M e^{-\epsilon_0P_0^{-1}}\\ 
\approx& 1 -e^{-\epsilon_0P_0^{-1}}+ \bar{P}^{-1}P_0  \sum^{M}_{n=1}\frac{M!}{n!(M-n)!}e^{-n\tau} \epsilon_{s,n}    \left(1-e^{-\tau}\right)^{M-n} e^{-\bar{\tau}_n  }\left( 1 +\bar{\tau}_n-\tau_n\right)    . \label{theorem2 ox} 
\end{align}
It is important to point out that $\epsilon_{s,n}$ is also a function of $n$ since $R_{sum,n}=nR_i$. 
We note that $\bar{\tau}_n= \tau_n $ since $\tau_n>\epsilon_0P_0^{-1}$ for $\bar{P}\rightarrow \infty$.  Therefore, $\mathbb{P}_0^{\mathrm{II, OL}} $ can be approximated as follows:
\begin{align}  
\mathbb{P}_0^{\mathrm{II, OL}}   
\approx& 1 -e^{-\epsilon_0P_0^{-1}}+ \bar{P}^{-1}P_0  \sum^{M}_{n=1}\frac{M!}{n!(M-n)!}e^{-n\tau}    \left(1-e^{-\tau}\right)^{M-n}   \epsilon_{s,n} e^{-\frac{n\tau\epsilon_{s,n}^{-1}\bar{P}-1}{P_0}  }  
\\\nonumber
=& 1 -e^{-\epsilon_0P_0^{-1}}+ e^{\frac{1}{P_0}  } \bar{P}^{-1}P_0  \sum^{M}_{n=1}\frac{M!}{n!(M-n)!}e^{-n\tau}   \left(1-e^{-\tau}\right)^{M-n}   \epsilon_{s,n} \left(e^{-\frac{\tau\epsilon_{s,n}^{-1}\bar{P}}{P_0}  }  \right)^n.
\end{align}
Since $\bar{P}\rightarrow \infty$, the approximation of $\mathbb{P}_0^{\mathrm{II, OL}} $ can be further simplified as follows:
\begin{align}  \label{approximation type 2}
\mathbb{P}_0^{\mathrm{II, OL}}   
\approx&  1 -e^{-\epsilon_0P_0^{-1}}+ \underset{\Delta{\mathbb{P}_0^{\mathrm{II, OL}} }}{\underbrace{\frac{e^{\frac{1}{P_0}  } P_0   Me^{-\tau}   \left(1-e^{-\tau}\right)^{M-1}   \epsilon_{s,1}     }{\bar{P}e^{\frac{\tau\epsilon_{s,1}^{-1}\bar{P}}{P_0}  }}}}.
\end{align}

{\it Remark 8:} The approximation in \eqref{approximation type 2} shows that the difference between the outage probabilities for the cases where  grant-free users are and are not admitted to $\mathrm{B}_0$, $\Delta{\mathbb{P}_0^{\mathrm{II, OL}} }$,  approaches zero, as $\bar{P}$ is increaseed and   $P_0$ is fixed. This is different from the behaviour of  the proposed Type I SGF  protocol, where the difference is reduced to zero if $P_0$ is increased and   $\bar{P}$ is fixed. This difference is due to the different SIC decoding orders employed  by the two protocols. 

{\it Remark 9:}  Compared to the case considered in Remark 2,      the case    with  small $P_0$ and  large $\bar{P}$  represents another important uplink scenario, where the grant-based user is a cell-edge user and the grant-free users  are close to the base station. The proposed Type II SGF protocol is ideally suited for such uplink communication scenarios, as  massive connectivity can be effectively supported and  the QoS requirement of the grant-based user can   be strictly guaranteed. 

{\it Remark 10:} Following   steps similar to \eqref{approximation type 2}, one can also show that the outage probability difference  $\Delta{\mathbb{P}_0^{\mathrm{II, OL}} }$ is reduced to zero by increasing $\tau$, instead of reducing $\tau$ which was considered in Remark 1.

%%%%%%%%%%%%%%%%%%%%%%%%%%%%%%%%%%%%%%

\subsubsection{Impact of SGF transmission on grant-free  users}
In order to analyze the impact of the proposed protocol on the grant-free   users' data rates, in this section, we assume   that $N$ grant-free  users have been selected. 
Similar to \eqref{rate set}, it is assumed that each user can adapt its transmit data rate as follows: 
\begin{align} \label{taxx2}
 \log\left(1 +\frac{|h_{(i)}|^2\bar{P}}{\sum^{i-1}_{j=M-N+1}|h_{(j)}|^2\bar{P}+|h_0|^2P_0+1}\right), 
 \end{align}
where $(M-N+1)\leq i\leq M$.  We focus on   the probability for  a  grant-free  user to successfully send $L$ bits  to the base station within $\mathrm{B}_0$, which is obtained as follows:
\begin{align}\label{Pi3}
\mathbb{P}_i^{\mathrm{II, OL}}  = \mathbb{P}\left(\mathrm{B}_0\log\left(1 +\frac{|h_{(i)}|^2\bar{P}}{\underset{j=M-n+1}{\sum^{i-1}} |h_{(j)}|^2\bar{P}+|h_0|^2P_0+1}\right)
<L\right),
\end{align}
 where the expression for the trivial case with   $i=M-N+1$ can be obtained similarly. 
 Probability $\mathbb{P}_i^{\mathrm{II, OL}} $ can be rewritten as follows:
\begin{align} \nonumber
\mathbb{P}_i^{\mathrm{II, OL}}  =& \mathbb{P}\left( \bar{P}\sum^{i-1}_{j=M-N+1}|h_{(j)}|^2+P_0|h_0|^2>\breve{\epsilon}_i^{-1}\bar{P}|h_{(i)}|^2-1 \right)\\\label{corela}=& \mathbb{P}\left( \bar{P}\sum^{i-1}_{j=M-N+1}|h_{(j)}|^2+P_0|h_0|^2>\breve{\epsilon}_i^{-1}\bar{P}|h_{(i)}|^2-1 , |h_{(i)}|^2>\breve{\epsilon}_i\bar{P}^{-1} \right) \\\nonumber &+\mathbb{P}\left(|h_{(i)}|^2<\breve{\epsilon}_i\bar{P}^{-1} \right). 
\end{align}
 Compared  to \eqref{Pidfd}, the probability in \eqref{corela} is more difficult to evaluate as  there are three random variables,  $|h_{0}|^2 $, $|h_{(i)}|^2 $, and $\sum^{i-1}_{j=M-N+1}|h_{(j)}|^2$, involved in the expression. The fact that    $|h_{(i)}|^2 $ and $\sum^{i-1}_{j=M-N+1}|h_{(j)}|^2$ are correlated makes the analysis more challenging. Therefore, we rely on computer simulations for the performance analysis, see Section  \ref{Section simulation}. 
 
  \subsection{Distributed Contention Control}  
After the base station broadcasts the threshold $\tau$,  assume that there are $N$ users whose channel gains  are stronger than  $\tau$ and only  these $N$ users are allowed to participate in distributed contention control. We note that the unordered channel gains of these users follow the distribution in \eqref{tildeg}.  The   user with the strongest channel gain is selected, instead of the weakest user as for the proposed Type~I SGF protocol.  Therefore, the outage probability of $\text{U}_0$ can be expressed as follows:
\begin{align}\label{Pxxz}
\mathbb{P}_0^{\mathrm{II, DCC}}  = &1 - \sum^{M}_{n=1}\mathbb{P}\left(|h_{(M-n)}|^2<\tau, |h_{(M-n+1)}|^2>\tau\right)Q_4\\\nonumber &-\mathbb{P}\left(|h_{(M)}|^2<\tau\right)  \mathbb{P}\left(\log\left( 1+|h_0|^2P_0 \right)>R_0\right),
\end{align}
where again it is assumed that all  grant-free  users use the same target data rate, $R_i$, and 
\begin{align}
Q_4=&\mathbb{P}\left(\log\left( 1+|h_0|^2P_0 \right)>R_0,  \left. \log\left(1+\frac{\bar{P}  |h_{(M)}|^2}{1+|h_0|^2P_0}\right)>R_{i}\right|N=n\right).
\end{align}
  $\mathbb{P}_0^{\mathrm{II, DCC}} $ can be simplified to the following equivalent form:
\begin{align}\label{Pxxza}
\mathbb{P}_0^{\mathrm{II, DCC}}  = &1 -  \underset{\tilde{Q}_{4}}{\underbrace{\mathbb{P}\left( |h_{(M)}|^2>\tau, |h_{(M)}|^2>\bar{P}^{-1}\epsilon_i(1+P_0|h_0|^2), |h_0|^2>P_0^{-1}\epsilon_0 \right)}}\\\nonumber &-\mathbb{P}\left(|h_{(M)}|^2<\tau\right)  \mathbb{P}\left(\log\left( 1+|h_0|^2P_0 \right)>R_0\right),
\end{align}

By analyzing the constraints of  $ |h_{(M)}|^2$ and $ |h_{0}|^2$, $\tilde{Q}_{4}$ can be expressed as follows:
\begin{align}\nonumber 
\tilde{Q}_{4}   =&  \mathbb{P}\left( |h_{(M)}|^2>\bar{P}^{-1}\epsilon_i(1+P_0|h_0|^2),  \bar{P}^{-1}\epsilon_i(1+P_0|h_0|^2)>\tau, |h_0|^2>P_0^{-1}\epsilon_0 \right)\\\nonumber 
&+\mathbb{P}\left(   |h_{(M)}|^2>\tau, \bar{P}^{-1}\epsilon_i(1+P_0|h_0|^2)<\tau, |h_0|^2>P_0^{-1}\epsilon_0 \right).
\end{align}
With some algebraic manipulations,   $\tilde{Q}_{4}$ can be evaluated as follows:
 \begin{align}  \nonumber
\tilde{Q}_{4} =&    e^{-\theta_h} - \sum^{M}_{k=0}{M\choose k}(-1)^k  e^{-k \bar{P}^{-1}\epsilon_i }\frac{e^{-(1+k \bar{P}^{-1}\epsilon_i P_0) \theta_h}}{(1+k \bar{P}^{-1}\epsilon_i P_0)}\\  &+ \left[1-\left(1-e^{-\tau}\right)^M \right]\left(e^{-P_0^{-1}\epsilon_0}-e^{-(P_0^{-1}\bar{P}\epsilon_i^{-1}\tau-P_0^{-1})}\right),
\end{align}
if $P_0^{-1}\bar{P}\epsilon_i^{-1}\tau-P_0^{-1}>P_0^{-1}\epsilon_0$, otherwise
 \begin{align}    
\tilde{Q}_{4} =&    e^{-\theta_h} - \sum^{M}_{k=0}{M\choose k}(-1)^k  e^{-k \bar{P}^{-1}\epsilon_i }\frac{e^{-(1+k \bar{P}^{-1}\epsilon_i P_0) \theta_h}}{(1+k \bar{P}^{-1}\epsilon_i P_0)},
\end{align}
where $\theta_h = \max(P_0^{-1}\epsilon_0, P_0^{-1}\bar{P}\epsilon_i^{-1}\tau-P_0^{-1})$. By substituting the expressions for $\tilde{Q}_{4}$ in \eqref{Pxxza} and   also using the fact that $\mathbb{P}\left(|h_{(M)}|^2<\tau\right)  \mathbb{P}\left(\log\left( 1+|h_0|^2P_0 \right)>R_0\right)=\left(1-e^{-\tau}\right)^Me^{-\epsilon_0P_0^{-1}}$, the outage probability $\mathbb{P}_0^{\mathrm{II, DCC}} $ can be obtained for the case with distributed contention control.

On the other hand, to analyze the performance of the selected grant-free  user,  assume that there is a fixed number of grant-free users, denoted by $N$, whose channel gains are above the threshold. Therefore, conditioned on $N$, the outage probability of the selected grant-free  user is given by
\begin{align}
\mathbb{P}_N^{\mathrm{II, DCC}}  =& \mathbb{P}\left(\log\left(1+\frac{\bar{P}  |h_{(N)}|^2}{1+|h_0|^2P_0}\right)<R_{i} \right).
\end{align} 
Note that $|h_{(N)}|^2>\tau$, which means that $\mathbb{P}_N^{\mathrm{II, DCC}} $ can be evaluated as follows:
\begin{align}  
\mathbb{P}_N^{\mathrm{II, DCC}}  =&  \mathbb{P}\left( |h_{(N)}|^2<\bar{P}^{-1}\epsilon_i(1+P_0|h_0|^2)  \right) \\\nonumber
=&\mathbb{P}\left( |h_{(N)}|^2<\bar{P}^{-1}\epsilon_i(1+P_0|h_0|^2) , \bar{P}^{-1}\epsilon_i(1+P_0|h_0|^2) >\tau \right) . 
\end{align}
Following   steps similar to those in the proof for Theorem~\ref{theorem2},  $\mathbb{P}_N^{\mathrm{II, DCC}}  $ can be expressed in closed form   as follows: 
\begin{align}  
\mathbb{P}_N^{\mathrm{II, DCC}}  =&  \sum^{N}_{k=0}{N \choose k} (-1)^ke^{k\tau} e^{-k \bar{P}^{-1}\epsilon_i}  \frac{e^{-(1+k \bar{P}^{-1}\epsilon_iP_0)  (P_0^{-1}\bar{P}\epsilon_i^{-1}\tau-P_0^{-1})^+}}{1+k \bar{P}^{-1}\epsilon_iP_0}
,
\end{align}
where recall $(x)^+\triangleq \max(0, x)$. 

{\it Remark 11:} Recall that the use of the proposed Type I SGF protocol with distributed contention control yields an error floor for the outage probability of the selected grant-free  user. This error floor does not exist for the Type II SGF protocol as explained in the following. Take the case where   $P_0$ and $\tau$ are fixed and $\bar{P}\rightarrow \infty$ as an example. One can find that  $\mathbb{P}_N^{\mathrm{II, DCC}} \rightarrow 0$ since
\begin{align}  
\mathbb{P}_N^{\mathrm{II, DCC}}  \rightarrow &  \sum^{N}_{k=0}{N \choose k} (-1)^ke^{k\tau} e^{-k \bar{P}^{-1}\epsilon_i}  e^{-P_0^{-1}\bar{P}\epsilon_i^{-1}\tau} 
\\\nonumber \rightarrow 0&  \left(1 -e^{\tau}   \right)^N   e^{-P_0^{-1}\bar{P}\epsilon_i^{-1}\tau}\rightarrow 0.
\end{align}
Note that  $\mathbb{P}_N^{\mathrm{II, DCC}} \rightarrow 0$  also holds when $\bar{P}\rightarrow \infty$, even if $\tau\rightarrow 0$, i.e., all  grant-free  users can participate in the contention.  
 
%%%%%%%%%%%%%%%%%%%%%
\section{Numerical Studies}\label{Section simulation}
In this section, the performance of the proposed SGF protocols is evaluated by using computer simulations, where the conventional grant-free and grant-based schemes are used as benchmarks to facilitate performance evaluation. In particular, the grant-free  scheme admits all  $M$ grant-free users to $\text{B}_0$, whereas $\text{B}_0$ is solely occupied by $\text{U}_0$ for the grant-based scheme. 

In Fig.~\ref{fig 1}, the impact of the proposed Type I SGF protocol with open-loop contention control on the outage probability of $\text{U}_0$ is shown as a function of the transmit SNR, $P_0$, where the   noise power is assumed to be normalized.  As can be observed from Fig. \ref{fig 1}, it is possible to ensure that  $\text{U}_0$ communicates  with the base station as if it solely occupied $\text{B}_0$, while additional grant-free users are admitted to $\text{B}_0$.  The difference between the outage probabilities  for the grant-based and the SGF schemes   is insignificant if the transmit power of the grant-free  users is small, i.e., $\bar{P}=0$ dB. This is because      the grant-free  users do not cause strong interference to $\text{U}_0$ if $\bar{P}$ is small. When the transmit power of the grant-free  users is large, i.e.,  $\bar{P}=20$ dB, it is important to reduce the value of   threshold, $\tau$. As such,  fewer  grant-free  users are admitted to $\text{B}_0$, and hence the outage performance of the proposed SGF protocol remains  similar to that of the grant-based scheme, as shown in Fig. \ref{fig1a}. On the other hand, the grant-free scheme results in the worst performance among the three considered  schemes since it admits all grant-free users and hence introduces severe  interference.  The two subfigures in Fig. \ref{fig 1} also demonstrate  that the developed analytical results perfectly match the simulation results, and the gap   between the asymptotic results and the simulation results is reduced for large  $P_0$. 

\begin{figure}[!htp] 
\begin{center}\subfigure[$\tau=0.1$]{\label{fig1b}\includegraphics[width=0.43\textwidth]{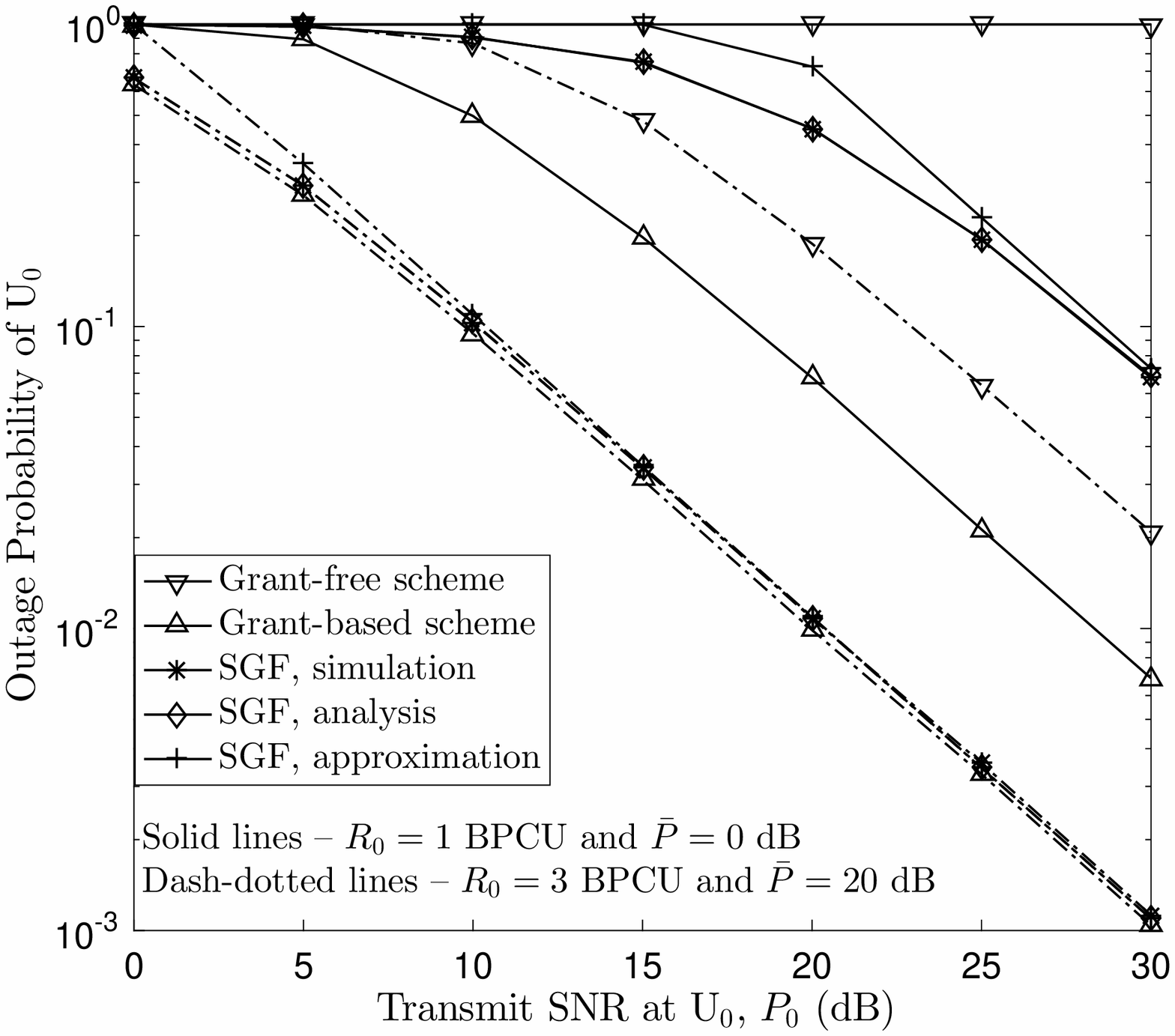}} \subfigure[ $\tau=\frac{1}{P_0}$ ]{\label{fig1a}\includegraphics[width=0.43\textwidth]{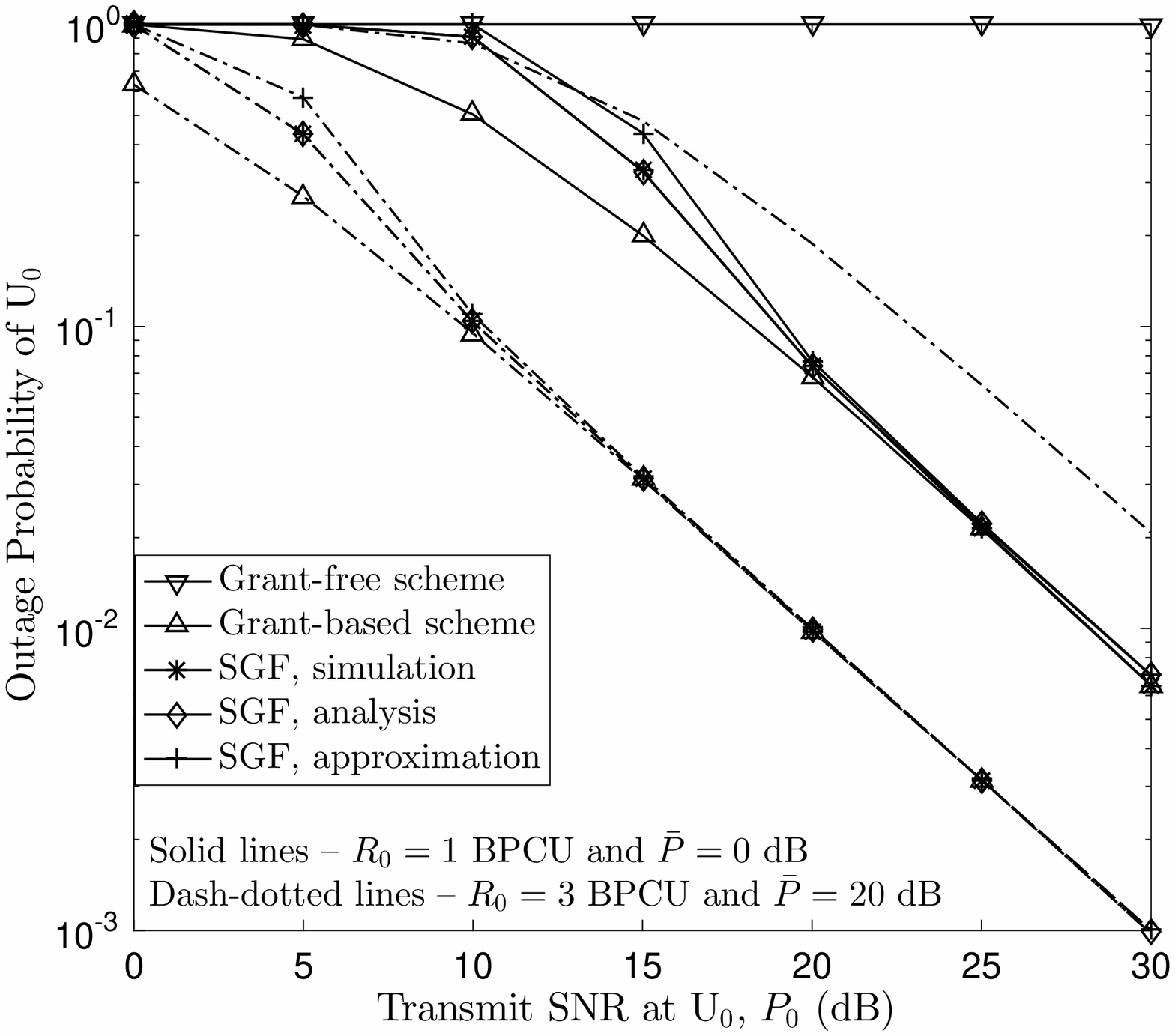}}
\end{center}\vspace{-1em}
 \caption{Impact of Type I open-loop  SGF NOMA transmission  on the outage probability of $\text{U}_0$. $M=20$. BPCU denotes bit per channel use.    }\label{fig 1}
\end{figure}

\begin{figure}[!htbp]\centering \vspace{-1em}
    \epsfig{file=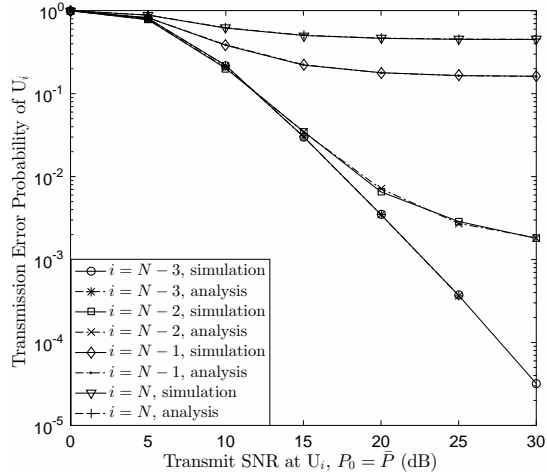, width=0.43\textwidth, clip=}\vspace{-1em}
\caption{ Impact of Type I open-loop  SGF NOMA transmission  on the transmission error probability of $\text{U}_i$. $N=5$, $\tau=0.5$, and $\breve{R}_i=0.6$ BPCU.   \vspace{-1em} }\label{fig2}
\end{figure}

In Fig. \ref{fig2}, the impact of the open-loop Type I SGF protocol   on the grant-free  users' data rates is studied, where it is assumed that there are $N=5$ selected   grant-free  users. The definition of the transmission errors is based on \eqref{Pi}. As can be observed from the figure,    the grant-free  users' transmission error  probabilities exhibit   error floors. This is due to the fact that $\text{U}_i$  is impaired by  the interference from   grant-free  users $\text{U}_j$, $j<i$. These error floors are  reduced by reducing $i$, since $\text{U}_i$ is affected by less interference than $\text{U}_k$, for $i<k$. Fig. \ref{fig2} also demonstrates the accuracy of the developed analytical results. 

\begin{figure}[!htbp]\centering \vspace{-1em}
    \epsfig{file=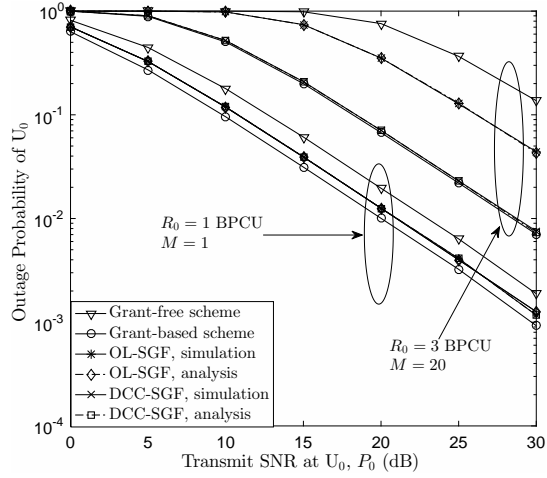, width=0.43\textwidth, clip=}\vspace{-1em}
\caption{ Impact of   Type I SGF NOMA transmission  on the outage probability of $\text{U}_0$ for the open-loop and distributed contention control protocols.  $\tau=1$ and $\bar{P}=0$ dB. OL stands for  open-loop and DCC stands for distributed contention control.  \vspace{-1em} }\label{fig3}
\end{figure}

\begin{figure}[!htbp]\centering \vspace{-1em}
    \epsfig{file=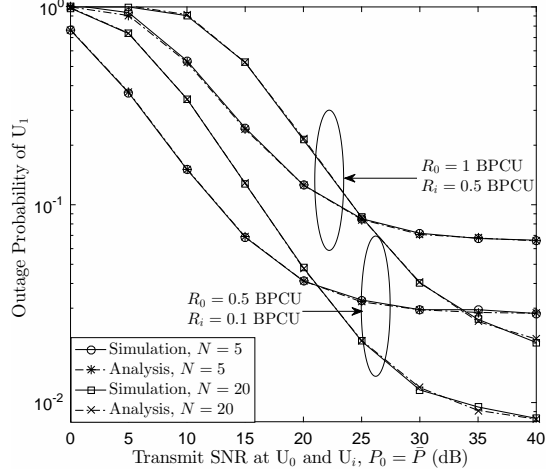, width=0.43\textwidth, clip=}\vspace{-1em}
\caption{ Impact of  Type I  SGF NOMA transmission with distributed contention control on the outage probability of the grant-free  user. $\tau=0.5$.   \vspace{-1em} }\label{fig4}
\end{figure}

In Figs. \ref{fig3} and \ref{fig4}, the performance of the proposed Type I SGF scheme with distributed contention control is evaluated. In Fig. \ref{fig3}, the performance of the Type I SGF scheme is studied  for two contention control protocols, open-loop contention control and  distributed contention control. When $M=1$, there is no difference between the two protocols, which is the reason why the   curves for the two protocols  overlap in the figure. For large  $M$,   distributed contention control ensures that $\text{U}_0$ experiences  almost the  same performance as the grant-based scheme, but   fewer grant-free users are scheduled compared to the open-loop scheme. It is worth pointing out that the grant-free scheme results in the worst performance for both cases. 
Fig. \ref{fig4} demonstrates the impact of the proposed SGF protocol with distributed contention control on the grant-free  user's outage probability. This figure reveals that  there is an error floor, which can be explained as follows. As can be observed from  \eqref{eq fig4}, the outage probability $\mathbb{P}_1^{\mathrm{I, DCC}} $ comprises  two outage events, $ \left(\log\left(1+\frac{|h_0|^2P_0}{ |h_{(1)}|^2\bar{P}+1}\right)>R_0,  \log\left(1+ |h_{(1)}|^2\bar{P}\right) <R_i \right)$ and $ \left(\log\left(1+\frac{|h_0|^2P_0}{ |h_{(1)}|^2\bar{P}+1}\right)<R_0\right)$. The event $ \left(\log\left(1+\frac{|h_0|^2P_0}{ |h_{(1)}|^2\bar{P}+1}\right)<R_0\right)$ is the   cause for the error floor since the SINR becomes a constant when both $P_0$ and $\bar{P}$ approach infinity, as shown in \eqref{remark 5sx}. 
This error floor can be reduced by increasing $N$ and reducing the users' rate requirements.   In particular,   a large $N$ can reduce the error floor because $|h_{(1)}|^2$ is more likely to be small   for large $N$ and hence the event $ \left(\log\left(1+\frac{|h_0|^2P_0}{ |h_{(1)}|^2\bar{P}+1}\right)<R_0\right)$ is less likely to happen. 

\begin{figure}[!htp] 
\begin{center}\subfigure[$\mathbb{P}_0^{\mathrm{II, OL}} $ for  $R_i=0.5$ BPCU and $R_0=2R_i$. ]{\label{fig5b}\includegraphics[width=0.43\textwidth]{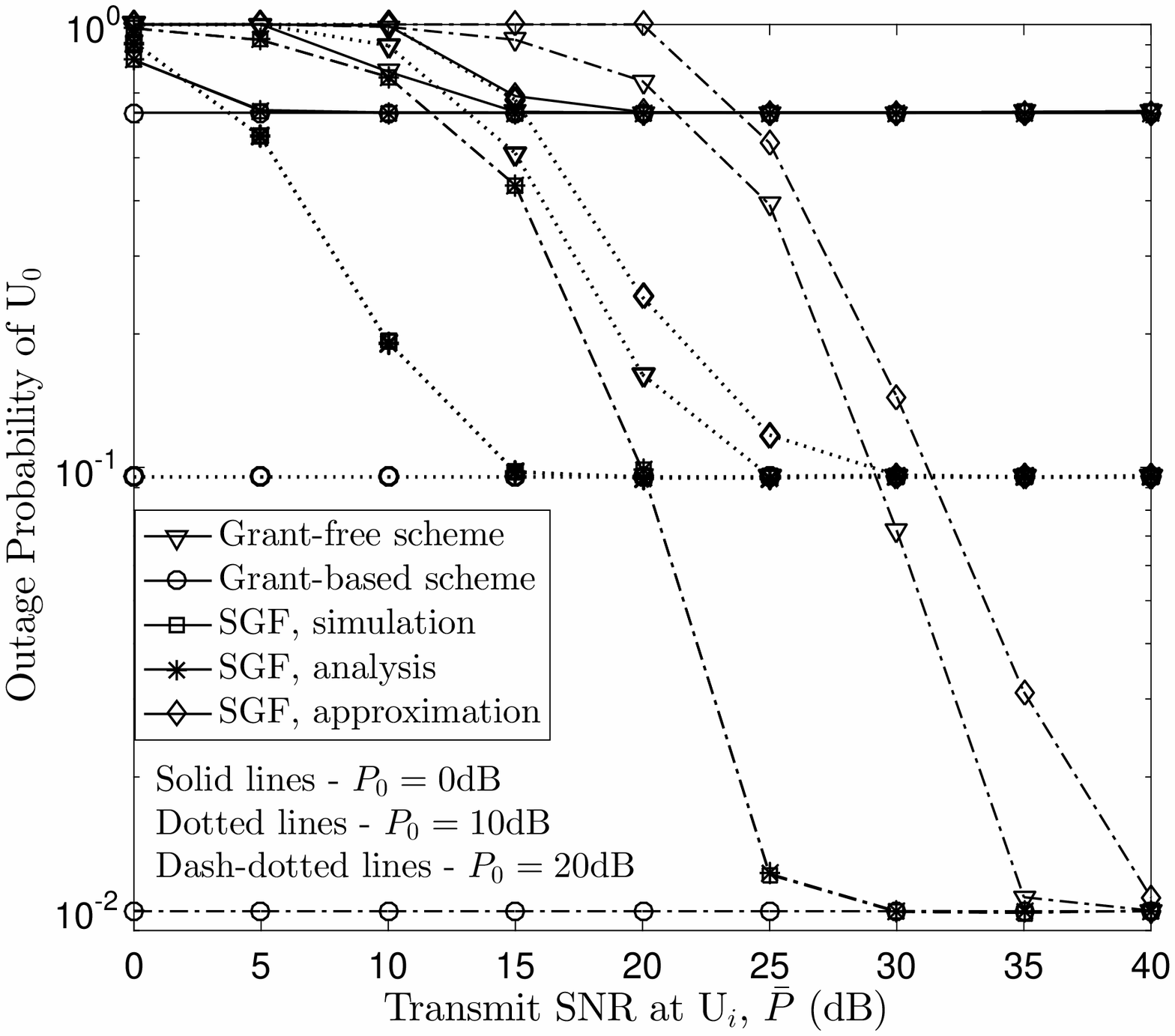}} \subfigure[ $\mathbb{P}_i^{\mathrm{II, OL}} $ for $P_0=20$ dB    and $N=5$. ]{\label{fig5a}\includegraphics[width=0.43\textwidth]{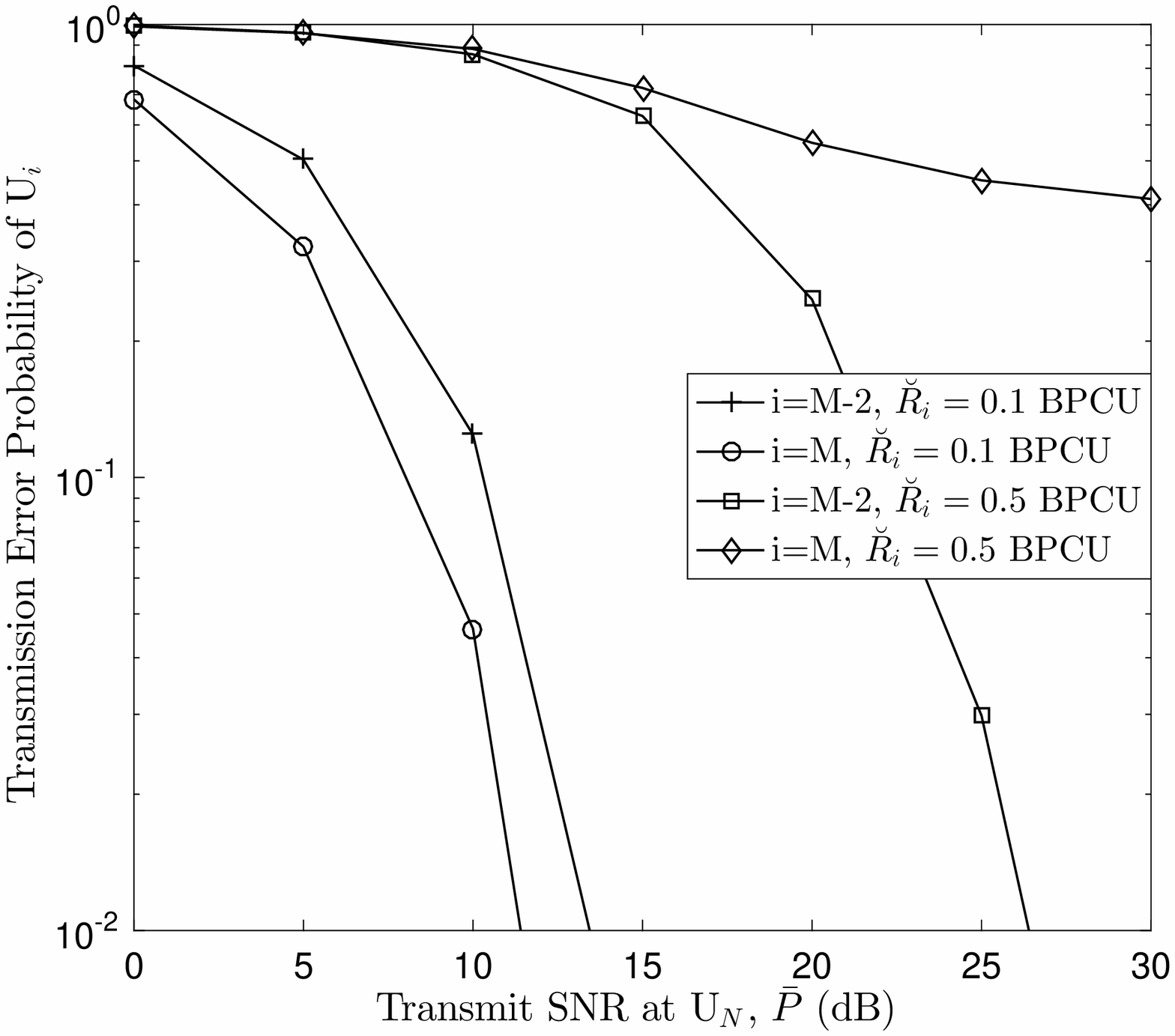}}
\end{center}
 \caption{Performance of  Type II  SGF NOMA transmission   with open-loop contention control. $M=10$ and $\tau=1$.  }\label{fig5}
\end{figure}

%\begin{figure}[!htbp]\centering \vspace{-1em}
%    \epsfig{file=type2_first.eps, width=0.43\textwidth, clip=}\vspace{-1em}
%\caption{ Impact of  Type II  SGF NOMA transmission  on the outage probability of $\mathrm{U}_0$ with open-loop contention control. $\tau=1$, $M=10$, $R_0=1$ BPCU and $R_i=0.5$ BPCU.   \vspace{-0.5em} }\label{fig5}
%\end{figure}

%\begin{figure}[!htbp]\centering \vspace{-1em}
%    \epsfig{file=OL_II_Pi.eps, width=0.43\textwidth, clip=}\vspace{-1em}
%\caption{ Impact of  Type II  SGF NOMA transmission  on the transmission error  probability of $\mathrm{U}_i$ with open-loop contention control. $\tau=0.1$, $M=10 $, and $N=5$.   \vspace{-0.5em} }\label{fig51}
%\end{figure}

\begin{figure}[!htbp]\centering \vspace{-1em}
    \epsfig{file=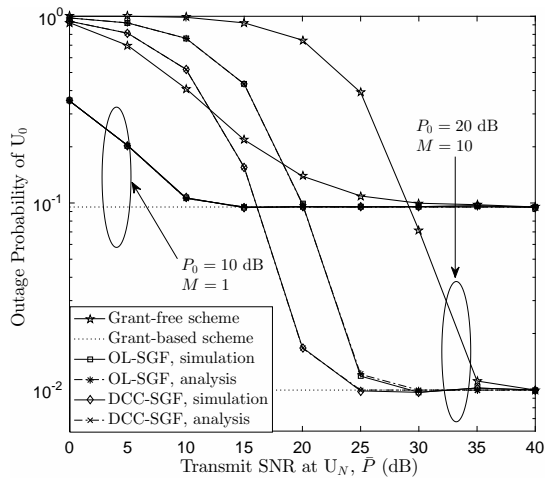, width=0.43\textwidth, clip=}\vspace{-1em}
\caption{ Impact of  Type II  SGF NOMA transmission  on the outage probability of $\mathrm{U}_0$ with open-loop contention control and distributed contention control. $\tau=1$, $M=20$, $R_0=1$ BPCU and $R_i=0.5$ BPCU.   \vspace{-2em} }\label{fig6}
\end{figure}

In Figs. \ref{fig5},  \ref{fig6}, and \ref{fig7}, the performance of the proposed Type II SGF protocol is evaluated, where the conventional  grant-based and grant-free protocols are used as benchmarks. Fig. \ref{fig5} demonstrates the impact of the proposed Type II SGF protocol with open-loop contention control on  $\mathbb{P}_0^{\mathrm{II, OL}} $ and $\mathbb{P}_i^{\mathrm{II, OL}} $, respectively.  Consistent  with Fig.~\ref{fig 1}, the use of the proposed   SGF protocol can ensure that $\text{U}_0$ experiences the same outage performance as if it solely occupied  $\text{B}_0$, while  the grant-free scheme realizes the worst performance among the three considered schemes. However, unlike the Type I protocol,  the gap between the grant-based and SGF schemes is reduced by increasing $\bar{P}$, as can be observed from   Fig. \ref{fig5b}.    The outage probability  in Fig. \ref{fig5b} has an error floor since the outage probability of the SGF scheme is lower bounded by the outage probability for the grant-based scheme, i.e., $\mathbb{P}(1+|h_0|^2P_0<R_0)$.   Fig. \ref{fig5a} shows the interesting phenomenon that  $\mathbb{P}_i^{\mathrm{II, OL}}$ is smaller than  $\mathbb{P}_j^{\mathrm{II, OL}} $, for $i>j$, if $\breve{R}_i$ is small enough. Otherwise,  $\mathbb{P}_i^{\mathrm{II, OL}}$ can be larger than  $\mathbb{P}_j^{\mathrm{II, OL}} $. Characterizing this impact of $\breve{R}_i$ on $\mathbb{P}_i^{\mathrm{II, OL}} $ by finding  a closed-form expression for  $\mathbb{P}_i^{\mathrm{II, OL}} $ is an important topic for future research.

In Fig. \ref{fig6}, the performances of the proposed Type II SGF schemes  with    open-loop and distributed contention control   are compared. Note that the two control mechanisms become identical  if there is only one grant-free  user, which is the reason why the  curves for the two schemes coincide  for $M=1$. By increasing $M$,  the open-loop based scheme offers the benefit that more grant-free  users are admitted to $\text{B}_0$, but the resulting  outage performance is worse than that of  the scheme with distributed contention control.  It is worth pointing out that both   schemes   outperform the grant-free scheme and achieve  the same performance as the grant-based scheme for large  $\bar{P}$. In Fig. \ref{fig7}, the impact of the Type II SGF scheme  with distributed contention  control on the selected grant-free  user's outage performance is studied. As discussed in Remark 11, unlike Type I SGF, the use of Type II SGF can ensure that the outage probability of the grant-free  user  approaches zero as $\bar{P}$ grows, which is confirmed by the simulation results   in Fig.~\ref{fig7}. 
 
\begin{figure}[!htbp]\centering \vspace{-1em}
    \epsfig{file=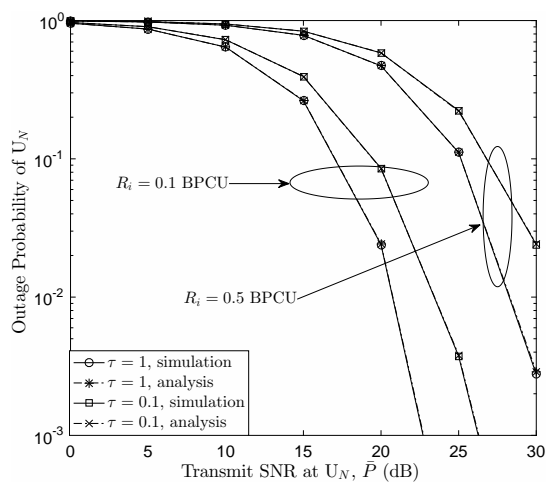, width=0.43\textwidth, clip=}\vspace{-1em}
\caption{ Impact of  Type II  SGF NOMA transmission  with    distributed contention control on the outage probability of $\mathrm{U}_N$.  $P_0=30$ dB and $N=5$.    \vspace{-1em} }\label{fig7}
\end{figure}

 \section{Conclusions}
In this paper, we have proposed  an SGF communication scheme, where one user is granted the right to transmit  via a grant based protocol and the other users are admitted to the same channel via a grant-free protocol. Two contention control mechanisms have been  proposed to ensure that the number of  users admitted to the same channel is carefully controlled. This feature  is particularly important  for scenarios  with an excessive number of active users, where    most existing grant-free schemes are not applicable  since  admitting a large number  of users to the same channel   can lead to  MUD failure. Analytical results and an asymptotic analysis have been provided   to demonstrate the superior performance of the proposed NOMA assisted SGF schemes and to study the impact of different SIC decoding orders. In particular, the proposed Type I SGF  schemes are ideally suited for the
scenario, where the grant-based user is close to the base station and the grant-free users are cell-edge users. On the other hand,  the proposed Type II SGF schemes   are ideal for the scenario, where the grant-based user is a cell-edge user and the grant-free users are close to the base station.  In this paper, each user is assumed to know its CSI perfectly.  An important topic for future research is the study of the impact of imperfect CSI   on the design of SGF schemes.
%%%%%%%%%%%%%%%%%%%%%%%%%%%
 \appendices
 \section{Proof for Theorem \ref{theorem1}} \label{appendix1}
 
By using order statistics \cite{David03}, $\mathbb{P}(N=n)$ can be obtained as follows: 
\begin{align}\label{N=n}
\mathbb{P}(N=n)=&\mathbb{P}\left(|h_{(n)}|^2<\tau, |h_{(n+1)}|^2>\tau\right)
\\\nonumber =&\frac{M!}{n!(M-n)!}e^{-(M-n)\tau}\left(1-e^{-\tau}\right)^n,
\end{align}
for $1\leq n \leq (M-1)$. It is straightforward to show that the expression in \eqref{N=n} is also valid for the case    $n=M$. The remainder  of the proof   focuses on the calculation of $Q_1$ in \eqref{q111}, the outage probability conditioned on $N=n$. 
  
$Q_1$ can be first rewritten    as follows:
\begin{align}
Q_1 =&  \mathbb{P}\left(\left. \sum^{n}_{j=1}|h_{(j)}|^2>\frac{ |h_0|^2P_0\bar{P}^{-1}}{2^{R_0}-1}-\bar{P}^{-1}\right| N=n\right).
\end{align}
The pdf of $\sum^{n}_{j=1}|h_{(j)}|^2$ can be found by treating it as the sum of the $n$ smallest order statistics among the $M$ channel gains \cite{Alam791}. Since all $|h_{(j)}|^2$, $1\leq j\leq n$, are smaller than $\tau$, a simple alternative is to treat $\sum^{n}_{j=1}|h_{(j)}|^2$ as the sum of $n$  i.i.d. variables, denoted by $|\tilde{h}_{j}|^2$  with the following CDF:
\begin{align}\label{fhj}
F_{|\tilde{h}_{j}|^2}(y)= \left\{ \begin{array}{ll}\frac{1-e^{-y}}{1-e^{-\tau}},
& \text{if}\quad y\leq \tau \\1,&\text{if}\quad y>\tau
 \end{array} \right..
\end{align} 
The  pdf of $|\tilde{h}_{j}|^2$, denoted by $f_{|\tilde{h}_{j}|^2}(y)$, can be obtained straightforwardly. The Laplace transform of $f_{|\tilde{h}_{j}|^2}(y)$ is given by
\begin{align}
\mathcal{L}\left(f_{|\tilde{h}_{j}|^2}(y)\right) = \frac{1-e^{-(s+1)\tau}}{(s+1)(1-e^{-\tau})}.
\end{align}
Since the $|\tilde{h}_{j}|^2$ are i.i.d., the Laplace transform of the pdf of $\sum^{n}_{j=1}|\tilde{h}_{j}|^2$ is given by
\begin{align}
\mathcal{L}\left(f_{\sum^{n}_{j=1}|\tilde{h}_{j}|^2}(y)\right) =& \frac{\left(1-e^{-(s+1)\tau}\right)^n}{(s+1)^n(1-e^{-\tau})^n}
\\\nonumber =&\sum^{n}_{p=0} \frac{{n \choose p}(-1)^pe^{-p\tau} }{(1-e^{-\tau})^n}\frac{e^{-p\tau s}}{(s+1)^n}.
\end{align}
By applying the inverse Laplace transform and also using the fact that $\sum^{n}_{j=1}|\tilde{h}_{j}|^2$ and $\sum^{n}_{j=1}| {h}_{(j)}|^2$ have the same pdf, the pdf of $\sum^{n}_{j=1}| {h}_{(j)}|^2$ is obtained as follows:
\begin{align}\nonumber
f_{\sum^{n}_{j=1}| {h}_{(j)}|^2}(y)=&\mathcal{L}^{-1}\left(\sum^{n}_{p=0} \frac{{n \choose p}(-1)^pe^{-p\tau} }{(1-e^{-\tau})^n}\frac{e^{-p\tau s}}{(s+1)^n}\right) \\\label{pdf sum} =& \sum^{n}_{p=0} \frac{{n \choose p}(-1)^pe^{-p\tau} }{(1-e^{-\tau})^n (n-1)!}    (y-p\tau)^{n-1} e^{-(y-p\tau)} u(y-p\tau) .
\end{align}
Intuitively, $f_{\sum^{n}_{j=1}| {h}_{(j)}|^2}(y)=0$ for $y\geq n\tau$, since $| {h}_{(j)}|^2<\tau$ for $1\leq j \leq n$. Because the pdf expression in \eqref{pdf sum} contains the step function, it is not straightforward  to show that   this  expression  fits the intuition. Thus, we verify this result  in the following lemma.
\begin{lemma} \label{lemma1}
The   expression for the pdf of $ \sum^{n}_{j=1}| {h}_{(j)}|^2 $ in \eqref{pdf sum} has the following property:
\begin{align}
f_{\sum^{n}_{j=1}| {h}_{(j)}|^2}(y)=0,
\end{align}
for $y\geq n\tau$. 
\end{lemma}
\begin{proof}
See Appendix \ref{appendix x1}.
\end{proof}

By using $f_{\sum^{n}_{j=1}| {h}_{(j)}|^2}(y)$,     probability   $Q_1$ can be expressed as follows:
\begin{align}
Q_1 =&\int^{\infty}_{\epsilon_0P_0^{-1}}\int_{\epsilon_x}^{\infty}f_{\sum^{n}_{j=1}| {h}_{(j)}|^2}(y)dy  f_{|h_0|^2}(x)dx +\int^{\epsilon_0P_0^{-1}}_0   f_{|h_0|^2}(x)dx\\\nonumber =& \sum^{n}_{p=0} \frac{{n \choose p}(-1)^pe^{-p\tau} }{(1-e^{-\tau})^n (n-1)!} \int^{\infty}_{\epsilon_0P_0^{-1}}\int_{\epsilon_x}^{\infty} (y-p\tau)^{n-1}\\\nonumber &\times  e^{-(y-p\tau)} u(y-p\tau) dyf_{|h_0|^2}(x)dx+1-e^{-\epsilon_0P_0^{-1}},
\end{align}
where $\epsilon_x=\frac{xP_0\bar{P}^{-1}}{2^{R_0}-1}-\bar{P}^{-1}$. 
Define $\bar{\epsilon}_x = (\epsilon_x-p\tau)^+$, where $(x)^+\triangleq \max(0, x)$.  

Note that the upper end of the integration range for $\sum^{n}_{j=1}| {h}_{(j)}|^2$ should be $n\tau$ since each ${h}_{(j)}$ is upper bounded by $\tau$, but   can be replaced by $\infty$ because of Lemma \ref{lemma1}.
Therefore,  $Q_1$ can be evaluated  as follows:
\begin{align}\nonumber 
Q_1   =& \sum^{n}_{p=0} \frac{{n \choose p}(-1)^pe^{-p\tau} }{(1-e^{-\tau})^n (n-1)!} \int^{\infty}_{\epsilon_0P_0^{-1}}  \int_{\bar{\epsilon}_x}^{\infty} x^{n-1} e^{-x}   dy\\\nonumber&\times f_{|h_0|^2}(x)dx +1-e^{-\epsilon_0P_0^{-1}}\\  &\overset{(b)}{=}\sum^{n}_{p=0} \frac{{n \choose p}(-1)^pe^{-p\tau} }{(1-e^{-\tau})^n (n-1)!}   \int^{\infty}_{\epsilon_0P_0^{-1}}  \Gamma(n,\bar{\epsilon}_x)f_{|h_0|^2}(x)dx +1-e^{-\epsilon_0P_0^{-1}},
\end{align}
where step $(b)$ follows  \cite[Eq. (3.38.3)]{GRADSHTEYN} and $\Gamma(\cdot)$ denotes the upper incomplete Gamma function. By using the series expansion  of the Gamma function \cite{GRADSHTEYN}, $Q_1$ can be expressed as follows:
\begin{align}
Q_1   =& 1-e^{-\epsilon_0P_0^{-1}}+  \sum^{n}_{p=0} \frac{{n \choose p}(-1)^pe^{-p\tau} }{(1-e^{-\tau})^n (n-1)!}  \sum_{l=0}^{n-1}\frac{(n-1)!  }{l!} \\\nonumber &\times \int^{\infty}_{\epsilon_0P_0^{-1}}  \left[ \left(\frac{xP_0\bar{P}^{-1}}{2^{R_0}-1}-\bar{P}^{-1}-p\tau\right)^+\right]^l e^{-(x+\bar{\epsilon}_x)}dx
\\\nonumber
=&  1-e^{-\epsilon_0P_0^{-1}}+ \sum^{n}_{p=0} \frac{{n \choose p}(-1)^pe^{-p\tau} }{(1-e^{-\tau})^n (n-1)!}  \sum_{l=0}^{n-1}\frac{(n-1)!  }{l!} \\\nonumber &\times \int^{\infty}_{\epsilon_0P_0^{-1}}  \left[ \left( \frac{x -\epsilon_0P_0^{-1}-\epsilon_0P_0^{-1} \bar{P}p\tau}{\epsilon_0P_0^{-1}\bar{P}}\right)^+\right]^l e^{-(x+\bar{\epsilon}_x)}dx.
\end{align}
The expression for $Q_1$ can be further simplified  as  follows:
\begin{align}\nonumber
Q_1 = 1-e^{-\epsilon_0P_0^{-1}}+  \sum^{n}_{p=0} \frac{{n \choose p}(-1)^pe^{-p\tau} }{(1-e^{-\tau})^n (n-1)!}  \left(\sum_{l=0}^{n-1}\frac{(n-1)!  }{l!} \right.\\\nonumber \times \underset{Q_{11}}{\underbrace{ \int^{\infty}_{\epsilon_0P_0^{-1}+\epsilon_0P_0^{-1} \bar{P}p\tau }  \left[  \frac{x -\epsilon_0P_0^{-1}-\epsilon_0P_0^{-1} \bar{P}p\tau}{\epsilon_0P_0^{-1}\bar{P}} \right]^l e^{-(x+\bar{\epsilon}_x)}dx}}\\ \nonumber \left.+ \underset{Q_{12}}{\underbrace{(n-1)!   \int_{\epsilon_0P_0^{-1}}^{\epsilon_0P_0^{-1}+\epsilon_0P_0^{-1} \bar{P}p\tau }   e^{-x }dx}}\right),
\end{align}
where the integral range in $Q_{11}$ is reduced since $\epsilon_0P_0^{-1}+\epsilon_0P_0^{-1} \bar{P}p\tau >\epsilon_0P_0^{-1}$ and $ \left( \frac{x -\epsilon_0P_0^{-1}-\epsilon_0P_0^{-1} \bar{P}p\tau}{\epsilon_0P_0^{-1}\bar{P}}\right)^+= 0$ if $x< \epsilon_0P_0^{-1}+\epsilon_0P_0^{-1} \bar{P}p\tau $.  $Q_{12}$ in the above equation is needed due to the fact that, for the  special case of $l=0$, the integral is not zero even if $\bar{\epsilon}_x=0$.

By applying   \cite[Eq. (3.381.4)]{GRADSHTEYN}, $Q_1$ can be expressed as follows:
\begin{align}
Q_1   =& \sum^{n}_{p=0} \frac{{n \choose p}(-1)^pe^{-p\tau} }{(1-e^{-\tau})^n (n-1)!} \left( \sum_{l=0}^{n-1}\frac{(n-1)!  }{l!}\right. \\\nonumber &\times \left.\frac{  e^{-\epsilon_0P_0^{-1}(1+  \bar{P}p\tau) } l!}{\epsilon_0^lP_0^{-l} \bar{P}^l (1+\epsilon_0^{-1}P_0 \bar{P}^{-1})^{(l+1)} } +(n-1)!\left(e^{-\epsilon_0P_0^{-1}}-e^{-\epsilon_0P_0^{-1}(1+\bar{P}p\tau)}\right)\right)   +1-e^{-\epsilon_0P_0^{-1}}.
\end{align}

Therefore, the outage probability of $\text{U}_0$ can be obtained as follows:
\begin{align}\nonumber
\mathbb{P}_0^{\mathrm{I, OL}}  =& \sum_{n=1}^{M}\frac{M!e^{-(M-n)\tau}\left(1-e^{-\tau}\right)^n}{n!(M-n)!}\left( \sum^{n}_{p=0} \frac{{n \choose p}(-1)^pe^{-p\tau} }{(1-e^{-\tau})^n }\right.\\\nonumber &\times \left. \left( \sum_{l=0}^{n-1} \frac{  e^{-\epsilon_0P_0^{-1}(1+  \bar{P}p\tau) } }{\epsilon_0^lP_0^{-l} \bar{P}^l (1+\epsilon_0^{-1}P_0 \bar{P}^{-1})^{(l+1)} } \right.\right.\\\nonumber &\left.\left.+ e^{-\epsilon_0P_0^{-1}}-e^{-\epsilon_0P_0^{-1}(1+\bar{P}p\tau)} \right)  +1-e^{-\epsilon_0P_0^{-1}}
\right)
\\\label{proof1} &+  e^{-M\tau} \left(1 - e^{-\epsilon_0P_0^{-1}}\right).
\end{align}
 Note that  
 \begin{align}\label{eq56}
  \sum_{n=1}^{M}\frac{M!}{n!(M-n)!}e^{-(M-n)\tau}\left(1-e^{-\tau}\right)^n  
+  e^{-M\tau}  =1.
\end{align}
By applying \eqref{eq56}   to \eqref{proof1}, the expression  for $\mathbb{P}_0^{\mathrm{I, OL}} $ can be simplified as shown in the theorem, and the proof is complete.  
 
 \section{Proof for Lemma \ref{lemma1}}\label{appendix x1}
The lemma can be proved by first  rewriting the pdf as follows:
\begin{align}
f_{\sum^{n}_{j=1}| {h}_{j}|^2}(y)  =&\frac{e^{-y} }{(1-e^{-\tau})^n (n-1)!} \sum^{n}_{p=0} {n \choose p}(-1)^p \\\nonumber &\times  (y-p\tau)^{n-1}  u(y-p\tau) .
\end{align}
In the case of $y\geq n\tau$, $u(y-p\tau) =1$, which means that the pdf can be simplified as follows:
\begin{align}
f_{\sum^{n}_{j=1}| {h}_{j}|^2}(y)  =&\frac{e^{-y}   }{(1-e^{-\tau})^n (n-1)!} \sum^{n}_{p=0} {n \choose p}(-1)^p   \sum^{n-1}_{l=0}y^{n-1-l}(-1)^lp^l \\\nonumber
=&\frac{e^{-y}   }{(1-e^{-\tau})^n (n-1)!} \sum^{n-1}_{l=0}y^{n-1-l}  \sum^{n}_{p=0} {n \choose p}(-1)^p(-1)^lp^l .
\end{align}
It is important to point out that $l$ is strictly smaller than $n$. According to  in \cite[Eq. (0.154.3)]{GRADSHTEYN}, we have 
\begin{align}
\sum^{n}_{p=0} {n \choose p}(-1)^p(-1)^lp^l =0,
\end{align}
for $l<n$. Therefore $f_{\sum^{n}_{j=1}| {h}_{j}|^2}(y) =0$, for  $y\geq n\tau$. This   completes the proof.

\section{Proof for Theorem \ref{theorem2}}\label{prooftheorem2}

The first step of the proof is to find $\mathbb{P}(N=n) $, which  can be expressed as follows:
\begin{align}\label{eq67}
\mathbb{P}(N=n) = \mathbb{P}\left(|h_{(M-n)}|^2<\tau, |h_{(M-n+1)}|^2>\tau\right).
\end{align}
Note that \eqref{eq67}   is different from \eqref{N=n}.  By applying order statistics,   the outage probability can be expressed as follows:
\begin{align} \label{p0xxx}
\mathbb{P}_0^{\mathrm{II, OL}}  = &1 - \sum^{M}_{n=1}\frac{M!}{n!(M-n)!}e^{-n\tau}\left(1-e^{-\tau}\right)^{M-n} Q_3\\\nonumber &-\left(1-e^{-\tau}\right)^M e^{-\epsilon_0P_0^{-1}}.
\end{align}
The remainder  of the proof is to find an expression for $Q_3$. 
With some algebraic manipulations, $Q_3$ can be expressed as follows:
\begin{align} 
Q_3 = & \underset{|h_0|^2>\frac{\epsilon_0}{P_0}}{\mathcal{E}}\left\{\mathbb{P}\left( \left. y_n>\epsilon_{s,n}\bar{P}^{-1}(1+P_0|h_0|^2)\right|N=n\right)\right\},
\end{align}
where $\mathcal{E}\{\cdot\}$ denotes the expectation operation, and  $y_n=\sum^{M}_{j=M-n+1} |h_{(j)}|^2 $. 

By applying order statistics, one can evaluate $Q_3$ by using the distribution of the sum of the $n$ largest order statistics among   $M $ Rayleigh fading gains. However, a simpler alternative is to treat $y_n$ as the sum of $n$ i.i.d. random variables, denoted by $\tilde{g}_j$, with the following CDF:
\begin{align}\label{tildeg}
F_{|\tilde{g}_{j}|^2}(y)= \left\{ \begin{array}{ll}0,
& \text{if}\quad y\leq \tau \\\frac{e^{-\tau}-e^{-y}}{e^{-\tau}},&\text{if}\quad y>\tau
 \end{array} \right.,
\end{align} 
which is different from \eqref{fhj} since $|\tilde{g}_j|^2\geq \tau$ and $|\tilde{h}_j|^2\leq \tau$. The pdf of $|\tilde{g}_{j}|^2$, $f_{|\tilde{g}_{j}|^2}(y)$, can be obtained straightforwardly.  The Laplace transform of $f_{|\tilde{g}_{j}|^2}(y)$ is given by
\begin{align}
\mathcal{L}\left(f_{|\tilde{g}_{j}|^2}(y)\right) = \frac{e^{-s\tau}}{s+1}.
\end{align}
Since the $|\tilde{g}_{j}|^2$ are i.i.d.,  the Laplace transform of the pdf of $y_n$ can be found as follows:
\begin{align}
\mathcal{L}\left(f_{y_n}(y)\right) = \frac{e^{-n\tau s}}{(s+1)^n},
\end{align}
which yields the following expression for the pdf of $y_n$
\begin{align}
f_{y_n} (y) = \frac{(y-n\tau)^{n-1} e^{-(y-n\tau)}u(y-n\tau)}{(n-1)!}. 
\end{align}

As a result, $Q_3 $ can be calculated as follow:
\begin{align} 
Q_3 = & \underset{|h_0|^2>\epsilon_0P_0^{-1}}{\mathcal{E}}\left\{ \int_{\epsilon_{s,n}\bar{P}^{-1}(1+P_0|h_0|^2)}^{\infty} f_{y_n}(y)dy\right\}\\\nonumber =& \underset{|h_0|^2>\epsilon_0P_0^{-1}}{\mathcal{E}}\left\{ \int_{(\Delta_h-n\tau)^+}^{\infty}  \frac{z^{n-1} e^{-z}}{(n-1)!}dy\right\} = \underset{|h_0|^2>\epsilon_0P_0^{-1}}{\mathcal{E}}\left\{    \frac{\Gamma(n,(\Delta_h-n\tau)^+)}{(n-1)!} \right\},
\end{align}
where $\Delta_h=  \epsilon_{s,n}\bar{P}^{-1}(1+P_0|h_0|^2) $.

Next, by applying the series expression of the incomplete Gamma function \cite{GRADSHTEYN}, we obtain the following: 
\begin{align} 
Q_3  =& \underset{|h_0|^2>\epsilon_0P_0^{-1}}{\mathcal{E}}\left\{     e^{-(\Delta_h-n\tau)^+} \sum^{n-1}_{l=0}\frac{\left[(\Delta_h-n\tau)^+\right]^l}{l!}\right\}\\\nonumber 
=& \int^{\infty}_{\epsilon_0P_0^{-1}}e^{-\epsilon_{s,n}\bar{P}^{-1}P_0\left(x-\frac{n\tau \epsilon_{s,n}^{-1}\bar{P}-1}{P_0}\right)^+} \sum^{n-1}_{l=0}\frac{\epsilon_{s,n}^l\bar{P}^{-l}P_0^l}{l!}   \left[\left(x-\frac{n\tau\epsilon_{s,n}^{-1}\bar{P}-1}{P_0}\right)^+\right]^l e^{-x}dx.
\end{align}
 
Similar to the proof for Theorem \ref{theorem1}, the result for the case  $l=0$ needs to be calculated separately   as follows:
\begin{align} \nonumber
Q_3  =& \int^{\infty}_{\epsilon_0P_0^{-1}}e^{-\epsilon_{s,n}\bar{P}^{-1}P_0\left(x-\frac{n\tau\epsilon_{s,n}^{-1}\bar{P}-1}{P_0}\right)^+} \sum^{n-1}_{l=1}\frac{\epsilon_{s,n}^l\bar{P}^{-l}P_0^l}{l!}   \left[\left(x-\frac{n\tau\epsilon_{s,n}^{-1}\bar{P}-1}{P_0}\right)^+\right]^l e^{-x}dx\\  &+\int^{\infty}_{\epsilon_0P_0^{-1}}e^{-\epsilon_{s,n}\bar{P}^{-1}P_0\left(x-\frac{n\tau\epsilon_{s,n}^{-1}\bar{P}-1}{P_0}\right)^+}  e^{-x}dx.
\end{align}
Since $\tau_n$ is not necessarily larger than $\epsilon_0P_0^{-1}$, we introduce  $\tau_n$ and $\bar{\tau}_n$ as defined in the theorem and $Q_3$ can be calculated as follows:
\begin{align} \nonumber 
Q_3  =&  \int^{\infty}_{\bar{\tau}_n}e^{-\epsilon_{s,n}\bar{P}^{-1}P_0\left(x-\tau_n\right)} \sum^{n-1}_{l=1}\frac{\epsilon_{s,n}^l\bar{P}^{-l}P_0^l}{l!}   \left[x-\tau_n\right]^l e^{-x}dx\\\nonumber &+ \int^{\bar{\tau}_n}_{\epsilon_0P_0^{-1}}    e^{-x}dx+\int^{\infty}_{\bar{\tau}_n}e^{-\epsilon_{s,n}\bar{P}^{-1}P_0\left(x-\tau_n\right)}      e^{-x}dx
\\\nonumber 
=&  \sum^{n-1}_{l=1}\frac{\epsilon_{s,n}^l\bar{P}^{-l}P_0^l}{l!}e^{-\tau_n} \frac{\Gamma\left(l+1, (\bar{\tau}_n-\tau_n)(1+\epsilon_{s,n}\bar{P}^{-1}P_0)\right)}{(1+\epsilon_{s,n}\bar{P}^{-1}P_0)^{l+1}}\\  \label{qx3}&+ e^{-\epsilon_0P_0^{-1}}-e^{-\bar{\tau}_n}  + \frac{  e^{- \tau_n- (\bar{\tau}_n-\tau_n)(1+\epsilon_{s,n}\bar{P}^{-1}P_0) }}{1+\epsilon_{s,n}\bar{P}^{-1}P_0}. 
\end{align}
By substituting \eqref{qx3} into \eqref{p0xxx}, the closed-form expression for $\mathbb{P}_0^{\mathrm{II, OL}} $ can be obtained and the proof is complete. 

\linespread{1.5}
   \bibliographystyle{IEEEtran}
\bibliography{IEEEfull,trasfer}

% Generated by IEEEtran.bst, version: 1.14 (2015/08/26)
\begin{thebibliography}{10}
\providecommand{\url}[1]{#1}
\csname url@samestyle\endcsname
\providecommand{\newblock}{\relax}
\providecommand{\bibinfo}[2]{#2}
\providecommand{\BIBentrySTDinterwordspacing}{\spaceskip=0pt\relax}
\providecommand{\BIBentryALTinterwordstretchfactor}{4}
\providecommand{\BIBentryALTinterwordspacing}{\spaceskip=\fontdimen2\font plus
\BIBentryALTinterwordstretchfactor\fontdimen3\font minus
  \fontdimen4\font\relax}
\providecommand{\BIBforeignlanguage}[2]{{%
\expandafter\ifx\csname l@#1\endcsname\relax
\typeout{** WARNING: IEEEtran.bst: No hyphenation pattern has been}%
\typeout{** loaded for the language `#1'. Using the pattern for}%
\typeout{** the default language instead.}%
\else
\language=\csname l@#1\endcsname
\fi
#2}}
\providecommand{\BIBdecl}{\relax}
\BIBdecl

\bibitem{7894280}
M.~Shafi, A.~F. Molisch, P.~J. Smith, T.~Haustein, P.~Zhu, P.~D. Silva,
  F.~Tufvesson, A.~Benjebbour, and G.~Wunder, ``{5G}: A tutorial overview of
  standards, trials, challenges, deployment, and practice,'' \emph{IEEE J. Sel.
  Areas Commun.}, vol.~35, no.~6, pp. 1201--1221, Jun. 2017.

\bibitem{nomama}
Z.~Ding, Y.~Liu, J.~Choi, Q.~Sun, M.~Elkashlan, C.-L. I, and H.~V. Poor,
  ``Application of non-orthogonal multiple access in {LTE} and {5G} networks,''
  \emph{IEEE Commun. Mag.}, vol.~55, no.~2, pp. 185--191, Feb. 2017.

\bibitem{jsacnomaxmine}
Z.~Ding, X.~Lei, G.~K. Karagiannidis, R.~Schober, J.~Yuan, and V.~Bhargava, ``A
  survey on non-orthogonal multiple access for {5G} networks: Research
  challenges and future trends,'' \emph{IEEE J. Sel. Areas Commun.}, vol.~35,
  no.~10, pp. 2181--2195, Oct. 2017.

\bibitem{6933459}
M.~Al-Imari, P.~Xiao, M.~A. Imran, and R.~Tafazolli, ``Uplink non-orthogonal
  multiple access for {5G} wireless networks,'' in \emph{Proc. 11th Int.
  Symposium on Wireless Commun. Systems (ISWCS)}, Barcelona, Spain, Aug 2014,
  pp. 781--785.

\bibitem{7557079}
M.~S. Ali, H.~Tabassum, and E.~Hossain, ``Dynamic user clustering and power
  allocation for uplink and downlink non-orthogonal multiple access {(NOMA)}
  systems,'' \emph{IEEE Access}, vol.~4, pp. 6325--6343, 2016.

\bibitem{7812683}
Y.~Sun, D.~W.~K. Ng, Z.~Ding, and R.~Schober, ``Optimal joint power and
  subcarrier allocation for full-duplex multicarrier non-orthogonal multiple
  access systems,'' \emph{IEEE Trans. Commun.}, vol.~65, no.~3, pp. 1077--1091,
  Mar. 2017.

\bibitem{6871674he}
H.~Chen, R.~Abbas, P.~Cheng, M.~Shirvanimoghaddam, W.~Hardjawana, W.~Bao,
  Y.~Li, and B.~Vucetic, ``Ultra-reliable low latency cellular networks: Use
  cases, challenges and approaches,'' to appear in 2018.

\bibitem{8345745}
X.~Sun, S.~Yan, N.~Yang, Z.~Ding, C.~Shen, and Z.~Zhong, ``Short-packet
  downlink transmission with non-orthogonal multiple access,'' \emph{IEEE
  Trans. Wireless Commun.}, vol.~17, no.~7, pp. 4550--4564, Jul. 2018.

\bibitem{Bennisurllc}
M.~Bennis, M.~Debbah, and H.~V. Poor, ``Ultra-reliable and low-latency
  communication: Tail, risk and scale,'' \emph{Proceedings of the IEEE}, to
  appear in 2018.

\bibitem{huawei}
``{5G}, a techology vision,'' Huawei, Inc., Shengzheng, China, 5G Whitepaper,
  Mar. 2015.

\bibitem{8395153}
X.~Bian, J.~Tang, H.~Wang, M.~Li, and R.~Song, ``An uplink transmission scheme
  for pattern division multiple access based on {DFT} spread generalized
  multi-carrier modulation,'' \emph{IEEE Access}, vol.~6, pp. 34\,135--34\,148,
  2018.

\bibitem{8352626}
G.~Ma, B.~Ai, F.~Wang, X.~Chen, Z.~Zhong, Z.~Zhao, and H.~Guan, ``Coded tandem
  spreading multiple access for massive machine-type communications,''
  \emph{IEEE Wireless Commun.}, vol.~25, no.~2, pp. 75--81, Apr. 2018.

\bibitem{8323218}
L.~Liu and W.~Yu, ``Massive connectivity with massive {MIMO} - {Part I}: Device
  activity detection and channel estimation,'' \emph{IEEE Trans. on Signal
  Process.}, vol.~66, no.~11, pp. 2933--2946, Jun. 2018.

\bibitem{8320821}
------, ``Massive connectivity with massive {MIMO} - {Part II}: Achievable rate
  characterization,'' \emph{IEEE Trans. Signal Process.}, vol.~66, no.~11, pp.
  2947--2959, Jun. 2018.

\bibitem{8482464}
Y.~Du, C.~Cheng, B.~Dong, Z.~Chen, X.~Wang, J.~Fang, and S.~Li,
  ``Block-sparsity-based multiuser detection for uplink grant-free {NOMA},''
  \emph{IEEE Trans. Wireless Commu.}, to appear in 2018.

\bibitem{7972955}
M.~Shirvanimoghaddam, M.~Condoluci, M.~Dohler, and S.~J. Johnson, ``On the
  fundamental limits of random non-orthogonal multiple access in cellular
  massive {IoT},'' \emph{IEEE J. Sel. Topics Signal Process.}, vol.~35, no.~10,
  pp. 2238--2252, Oct. 2017.

\bibitem{jsacnoma10}
J.~Choi, ``{NOMA} based random access with multichannel {ALOHA},'' \emph{IEEE
  J. Sel. Areas Commun.}, vol.~PP, no.~99, pp. 1--1, 2017.

\bibitem{viswanath02}
P.~Viswanath, D.~Tse, and R.~Laroia, ``Opportunistic beamforming using dumb
  antennas,'' \emph{IEEE Trans. Inform. Theory}, vol.~48, pp. 1277-- 1294, Jun.
  2003.

\bibitem{Zhiguo_mmwave}
Z.~Ding, P.~Fan, and H.~V. Poor, ``Random beamforming in millimeter-wave {NOMA}
  networks,'' \emph{IEEE Access}, vol.~5, pp. 7667--7681, 2017.

\bibitem{Zhao2005s}
Q.~Zhao and L.~Tong, ``Opportunistic carrier sensing for energy-efficient
  information retrieval in sensor networks,'' \emph{EURASIP Journal on Wireless
  Communications and Networking}, vol.~2, no.~2, pp. 1--1, Apr. 2005.

\bibitem{Bletsas06}
A.~Bletsas, A.~Khisti, D.~P. Reed, and A.~Lippman, ``A simple cooperative
  diversity method based on network path selection,'' \emph{IEEE Journal
  Select. Areas in Comm.}, vol.~24, pp. 659--672, Mar. 2006.

\bibitem{6334506}
R.~Talak and N.~B. Mehta, ``Feedback overhead-aware, distributed, fast, and
  reliable selection,'' \emph{IEEE Trans. Commu.}, vol.~60, no.~11, pp.
  3417--3428, Nov. 2012.

\bibitem{Zhiguo_CRconoma}
Z.~Ding, P.~Fan, and H.~V. Poor, ``Impact of user pairing on {5G}
  non-orthogonal multiple access,'' \emph{IEEE Trans. Veh. Tech.}, vol.~65,
  no.~8, pp. 6010--6023, Aug. 2016.

\bibitem{David03}
H.~A. David and H.~N. Nagaraja, \emph{Order Statistics}.\hskip 1em plus 0.5em
  minus 0.4em\relax John Wiley, New York, 3rd ed., 2003.

\bibitem{Alam791}
K.~Alam and K.~T. Wallenius, ``Distribution of a sum of order statistics,''
  \emph{Scandinavian Journal of Statistics}, vol.~6, no.~3, pp. 123--126, 1979.

\bibitem{GRADSHTEYN}
I.~S. Gradshteyn and I.~M. Ryzhik, \emph{Table of Integrals, Series and
  Products}, 6th~ed.\hskip 1em plus 0.5em minus 0.4em\relax New York: Academic
  Press, 2000.

\end{thebibliography}
   \end{document}